\documentclass{article}
\usepackage[utf8]{inputenc}
\usepackage{geometry}
\usepackage{amsmath,amsthm,amssymb}
\usepackage{graphics,graphicx,tikz,color,float, hyperref}
\usepackage{mathtools}
\usepackage{amsfonts}
\usepackage{tikz}
\usepackage {comment}
\usepackage{framed}
\usepackage[section]{placeins}
\usepackage{algorithm}
\usepackage[noend]{algpseudocode}
\usepackage{xcolor}
\newcommand{\mycomment}[1]{\textup{{\color{red}#1}}}

\newcommand{\rk}{\textsf{rank}}
\newtheorem{fact}{Fact}
\newtheorem{definition}{Definition}
\newtheorem{lemma}{Lemma}
\newtheorem{theorem}{Theorem}

\newcommand{\beq}{\begin{equation}}
	\newcommand{\enq}{\end{equation}}
\newcommand{\bel}{\begin{lemma}}
	\newcommand{\enl}{\end{lemma}}
\newcommand{\bet}{\begin{theorem}}
	\newcommand{\ent}{\end{theorem}}

\newcommand{\tr}{\mathrm{Tr}}

\newcommand{\E}{\mathbb{E}}
\newcommand{\ketbra}[1]{|#1\rangle\langle#1|}

\newcommand{\eps}{\varepsilon}

\newcommand*{\cH}{\mathcal{H}}

\newcommand*{\cO}{\mathcal{O}}

\newcommand*{\cX}{\mathcal{X}}

\newcommand*{\cZ}{\mathcal{Z}}
\newcommand*{\cE}{\mathcal{E}}

\newcommand*{\IP}{\mathsf{IP}}
\newcommand{\Ext}{\mathsf{Ext}}

\newcommand{\pre}{\mathsf{Crop}}

\newcommand{\advcb}{\mathsf{AdvCB}}
\newcommand{\ff}{\mathsf{FF}}
\newcommand{\ecc}{\mathsf{ECC}}

\newcommand{\supp}{\mathrm{supp}}
\newcommand{\suppress}[1]{}

\newcommand{\defeq}{\ensuremath{ \stackrel{\mathrm{def}}{=} }}
\newcommand{\F}{\mathbb{F}}

\newcommand {\br} [1] {\ensuremath{ \left( #1 \right) }}

\newcommand {\minusspace} {\: \! \!}
\newcommand {\smallspace} {\: \!}
\newcommand {\fn} [2] {\ensuremath{ #1 \minusspace \br{ #2 } }}

\newcommand {\dmax} [2] {\fn{\mathrm{D}_{\max}}{#1 \middle\| #2}}

\newcommand {\mutinf} [2] {\fn{\mathrm{I}}{#1 \smallspace : \smallspace #2}}

\newcommand {\condmutinf} [3] {\mutinf{#1}{#2 \smallspace \middle\vert \smallspace #3}}
\newcommand {\hminone} [1] {\fn{ \mathrm{H }_{\min}}{#1}}

\newcommand {\hmin} [2] {\fn{ \mathrm{H }_{\min}}{#1 \middle | #2}}

\newcommand {\id} {\ensuremath{\mathbb{I}}}

\newcommand {\Hmin}{\mathrm{H}_{\min}}

\newcommand{\nmextc}{\mathsf{2nmext \mhyphen c}}
\newcommand{\newnmextc}{\mathsf{new \mhyphen 2nmext \mhyphen c}}
\newcommand{\nmextq}{\mathsf{2nmext \mhyphen q}}
\newcommand{\newnmextq}{\mathsf{new \mhyphen 2nmext \mhyphen q}}
\newcommand{\IExt}{\mathsf{IExt}}
\newcommand{\Trev}{\mathsf{Trev}}

\newcommand{\samp}{\mathsf{Samp}}

\usepackage{afterpage}

\newcommand*{\cp}{\mathsf{copy}}
\newcommand*{\sm}{\mathsf{same}}

\newcommand*{\qmas}{\mathsf{qma\mhyphen state}}
\newcommand*{\nmas}{\mathsf{qnm\mhyphen state}}

\newcommand*{\nma}{\mathsf{qnm\mhyphen adv}}

\newcommand*{\nmext}{\mathsf{nmExt}}

\newcommand{\X}{\mathcal{X}}
\newcommand{\Y}{\mathcal{Y}}
\newcommand{\Z}{\mathcal{Z}}

\newcommand*{\cL}{\mathcal{L}}

\newcommand{\bra}[1]{\langle #1|}
\newcommand{\ket}[1]{|#1 \rangle}

\mathchardef\mhyphen="2D

\newcommand*{\enc}{\mathrm{Enc}}
\newcommand*{\dec}{\mathrm{Dec}}

\makeatletter
\newcommand*{\rom}[1]{\expandafter\@slowromancap\romannumeral #1@}
\makeatother

\mathchardef\mhyphen="2D

\newtheorem*{remark}{Remark}
\newtheorem{claim}{Claim}
\newtheorem{corollary}{Corollary}

\newcommand{\dnote}[1]{\marginpar{\textcolor{cyan}{\textbf{!!!}}}\textcolor{cyan}{\sf #1 --D.A.}}

\newfloat{Protocol}{tbhp}{lop}

\title{Quantum secure non-malleable codes in the split-state model}
\author{
	Divesh Aggarwal~\footnote{Center for Quantum Technologies, National University of Singapore, \texttt{dcsdiva@nus.edu.sg}}\qquad
	Naresh Goud Boddu~\footnote{Center for Quantum Technologies, National University of Singapore, \texttt{e0169905@u.nus.edu}}
	\and
	Rahul Jain~\footnote{Centre for Quantum Technologies and Department of Computer Science, 
  National University of Singapore and MajuLab, UMI 3654, Singapore,  \texttt{rahul@comp.nus.edu.sg}}
}
\date{}
\begin{document}
\maketitle
\begin{abstract}
{\em Non-malleable codes} introduced by Dziembowski, Pietrzak and Wichs~\cite{DPW10} encode a classical message $S$ in a manner such that the tampered codeword either decodes to the original message $S$ or a message that is unrelated/independent of $S$. Constructing non-malleable codes for various tampering function families has received significant attention in the recent years. We consider the well studied ($2$-part) {\em split-state}  model, in which the message $S$ is encoded into two parts $X$ and $Y$, and the adversary is allowed to arbitrarily tamper with each $X$ and $Y$ individually. Non-malleable codes in the split-state model have found applications in other important security notions like {\em non-malleable commitments} and {\em non-malleable secret sharing}. Thus, it is vital to understand if such non-malleable codes are secure against quantum adversaries.

We consider the security of non-malleable codes in the split-state model when the adversary is allowed to make use of arbitrary entanglement to tamper the parts $X$ and $Y$. We construct explicit quantum secure non-malleable codes in the split-state model. Our construction of quantum secure non-malleable codes is based on the recent construction of quantum secure {\em $2$-source non-malleable extractors} by Boddu, Jain and Kapshikar~\cite{BJK21}.
    \begin{itemize}
    \item We extend the connection of  Cheraghchi and Guruswami~\cite{CG14a}  between $2$-source non-malleable extractors and non-malleable codes in the split-state model in the classical setting to the quantum setting, i.e. we show that explicit quantum secure $2$-source  non-malleable extractors in $(k_1,k_2)\mhyphen\qmas$ framework of~\cite{BJK21} give rise to explicit quantum secure non-malleable codes in the split-state model. 
        \item We construct the first quantum secure non-malleable code with efficient encoding and decoding procedures for message length $m=n^{\Omega(1)}$, error $\eps=2^{-n^{\Omega(1)}}$ and codeword of size $2n$. Prior to this work, it remained open to provide such quantum secure non-malleable code even for a single bit message in the  split-state model.
    
\item We also study its  natural extension when the tampering of the codeword is performed $t$-times. 
         We construct quantum secure one-many non-malleable code with efficient encoding and decoding procedures for $t=n^{\Omega(1)}$, message length $m=n^{\Omega(1)}$, error $\eps=2^{-n^{\Omega(1)}}$ and codeword of size $2n$.
         \item As an application, we also construct the first quantum secure $2$-out-of-$2$ non-malleable secret sharing scheme for message/secret length $m=n^{\Omega(1)}$, error $\eps=2^{-n^{\Omega(1)}}$ and share of size $n$.
    \end{itemize}

\end{abstract}
\newpage
\section{Introduction}
In a seminal work, Dziembowski,
Pietrzak and Wichs~\cite{DPW10} introduced non-malleable codes to provide a meaningful guarantee for the encoded message $S$ in situations where traditional error-correction
or even error-detection is impossible. Informally, non-malleable codes encode a classical message $S$ in a manner such that tampering the codeword results in decoder either outputting the original message $S$ or a message that is unrelated/independent of $S$. Using probabilistic arguments,~\cite{DPW10} showed the existence of such non-malleable codes against any
family $\mathcal{F}$ of tampering functions  of size as large as $2^{2^{\alpha n}}$ for any fixed constant $\alpha<1$, where $n$ is the length of codeword for messages of length $\Omega(n)$.

Subsequent works continued to study non-malleable codes in various tampering models. Perhaps the most well known of these tampering function families is the so called split-state model introduced by Liu and Lysyanskaya~\cite{LL12}, who constructed efficient constant rate non-malleable codes against computationally bounded adversaries under strong cryptographic assumptions. We refer to the ($2$-part) split-state model as the split-state model in this paper. In the  split-state model, the message $S$ is encoded into two parts, $X$ and $Y$, after which the adversary is allowed to arbitrarily tamper $(X,Y) \to (X',Y')$ such that $(X',Y') =(f(X),g(Y))$ for any functions $(f,g)$ such that $f,g: \{0,1 \}^n \to  \{0,1 \}^n$. Dziembowski, Kazana and Obremski~\cite{DKO13} proposed a construction that provides non-malleable codes for a single
bit message based on strong extractors. Subsequently, multiple works considered non-malleable codes for multi-bit messages leading to constant rate non-malleable codes in the split-state model~\cite{LL12,CG14a,CGL15,CG14b,ADL17,A15,li15,AB16,Li17,Li19,AO20,AKOOS22}. Non-malleable codes in the split-state model have found applications to numerous other  important security notions such as non-malleable commitments and non-malleable secret sharing~\cite{GPR16,GK16,GK18,ADNOP19,AP19}.

More formally, a non-malleable code in the split-state model in the classical setting can be defined as follows. Let $n,m$  represent positive integers and $ k, \eps, \eps' > 0$ represent reals. Let $\mathcal{F}$ denote the set of all functions $f : \{0,1 \}^n \to \{0,1 \}^n$. We consider an encoding and decoding scheme $(\enc,\dec)$ in the split-state model where $\enc(S)=(X,Y)$. Here $S \sim U_m$ ($U_m$ is uniform distribution on $m$ bits) represents the plaintext/message and  $X,Y\in\{0,1\}^n$ are the two parts of the codeword. $\enc$ is a randomized function and $\dec(X,Y)$ is a deterministic function, such that $\Pr\left(\dec\left( \enc(S)\right) =S \right)=1$.

 We define the function $\cp : \{ 0,1\}^* \cup \{ \mathsf{same}\} \times \{ 0,1\}^* \to \{ 0,1\}^*$.  $\cp(x,s) =s$ if $x=\mathsf{same}$, otherwise $\cp(x,s) =x$. 

\begin{definition}[Non-malleable codes in the split-state model\label{def:nmcodes_classical}~\cite{LL12,DPW10}]$(\enc, \dec)$ is an $(m,n,\eps)$-non-malleable code with respect to a family of tampering functions $\mathcal{F} \times \mathcal{F}$, if for every $f=(g_1,g_2) \in \mathcal{F} \times \mathcal{F}$, there exists a random variable $\mathcal{D}_{f}=\mathcal{D}_{(g_1,g_2)}$ on $\{0,1 \}^m  \cup \lbrace \sm \rbrace$ which is independent of the randomness in $\enc$ such that for all messages $s \in \{0,1 \}^m$, it holds that
$$  \Vert  \dec(f (\enc(s)))  - \cp (\mathcal{D}_{f} ,s)\Vert_1 \leq \eps.$$
\end{definition}

Intuitively, if the adversary doesn't tamper the codeword (in which case $(X,Y)=(X',Y')$), the decoded message is same (captured by the variable $\mathsf{same}$) as original message $S$. If the adversary does tamper the codeword (in which case either $X \ne X'$ or $Y \ne Y'$), the decoded message is (approximately) distributed according to a distribution ($\mathcal{D}_{f}$ minus $\mathsf{same}$ normalized) that only depends on $f$ and is independent of the original message $S$.

\subsubsection*{Previous classical results in the split-state model}

\cite{DKO13}~constructed the first non-malleable code for a 1-bit message.  Following that Aggarwal, Dodis and Lovett~\cite{ADL17} gave
the first information-theoretic construction for $m$-bit messages, but the length of codeword being $2n=m^{\cO(1)}$. Chattopadhyay, Goyal and Li~\cite{CGL15} gave a non-malleable code for message length $m=n^{\Omega(1)}$, error $\eps=2^{- n^{\Omega(1)}}$ and codeword of size $2n.$ Improving upon the work of~\cite{CGL15}, Li~\cite{Li19} gave a non-malleable code for message length $m= \cO \left(\frac{n \log \log n }{\log n}\right)$, error $\eps=2^{- n^{\Omega(1)}}$ and codeword of size $2n.$ Only recently Aggarwal and Obremski~\cite{AO20} gave the first constant rate non-malleable code for message length $m=\Omega(n)$, error $\eps=2^{- n^{\Omega(1)}}$ and codeword of size $2n.$ This construction was improved to a rate $1/3$ construction in~\cite{AKOOS22}.

\suppress{
\dnote{The following should go in the "Our results" section, and should appear after introducing the quantum model, and motivating it}
Our work provides the first quantum secure non-malleable code with efficient encoding and decoding procedures for message length $m=n^{\Omega(1)}$, error $\eps=2^{- n^{\Omega(1)}}$ and codeword of size $2n$. When the tampering of the codeword is performed $t$-times, we also provide the first quantum secure one-many non-malleable code with efficient encoding and decoding procedures for $t=n^{\Omega(1)}$, message length $m=n^{\Omega(1)}$, error $\eps=2^{-n^{\Omega(1)}}$ and codeword of size $2n$. Prior to our work, it remained open to provide such quantum secure non-malleable codes even for a single bit message in the  split-state model.

}

\subsubsection*{Motivation to consider the  quantum setting}
 With the rise of quantum computers, it becomes vital to understand if non-malleable codes are secure against quantum adversaries. Quantum entanglement between various parties, used to generate classical information introduces {\em non-local correlations}~\cite{Bell64}. For example in the CHSH game, one can use local measurements on both the halves of a EPR state to generate a probability distribution which contains correlations stronger than those possible classically. Entanglement is of course known to yield several such unexpected effects with no classical counterparts,~e.g.,~{\em super-dense coding}~\cite{CS92}. Thus, it motivates us to consider if one can provide non-malleable codes when adversary in the  split-state model is allowed to make use of an arbitrary entanglement (between the two parts) to tamper the two parts $X$ and $Y$ (both classical) of an  encoded message $S$. 
 We note to the reader that the $(\enc,\dec)$ schemes considered in this paper are classical, and we provide quantum security in the sense that the adversary is allowed to do quantum operations to tamper $(X,Y) \to (X',Y')$ using pre-shared unbounded entanglement.

 \subsubsection*{Our results}

Our first contribution is setting up the required analogue/framework to define non-malleable codes in the quantum setting.

\subsubsection*{Quantum split-state adversary}
To tamper $(X,Y) \to (X',Y')$, we let the adversary share an arbitrary entanglement ${\psi}_{NM}$~\footnote{Without any loss of generality, one can consider $\psi$ as a pure state.} between the two different locations where split codewords are stored. The adversary then applies isometries  $U: \cH_X \otimes \cH_N \rightarrow  \cH_{X^\prime}  \otimes \cH_{N^\prime}$ and $V: \cH_Y \otimes \cH_M \rightarrow \cH_{Y'} \otimes \cH_{M'}$ to tamper~\footnote{This is without any loss of generality since one can consider Strinespring isometry extension~(see Definition~\ref{fact:stinespring}) of a CPTP map $\Phi$ as isometry if the tampering is given by a CPTP map $\Phi$.} $(X,Y) \to (X',Y')$. The decoding process begins by first measuring $(X^\prime, Y^\prime)$ and then outputting the decoded message $S'$ from $(X',Y')$ (post measurement in the computational basis). To show that non-malleable codes are secure against such an adversary, it is sufficient to show that if the adversary doesn't tamper the codeword, the decoded message $S'$ is same as the original message $S$. If the adversary does tamper the codeword, the decoded message $S'$ is (approximately) distributed according to a distribution~($\mathcal{D}_{(U,V,{\psi})}$ that only depends on $(U,V,{\psi})$) that is independent of the original message $S$ . For simplicity, we denote quantum split-state adversary as $\mathcal{A}=(U,V,{\psi}_{})$ in this paper.

We now formally define a quantum split-state adversary in the split-state model.

\begin{definition}[Quantum split-state adversary (see Figure~\ref{fig:splitstate1})]\label{def:splitstateadv}
      Let $\sigma_{SXY}$ be the state after encoding the message $S$. The quantum split-state adversary (denoted $\mathcal{A} = (U,V,{\psi}_{})$) will act via two isometries, $(U,V)$ using an additional shared entangled state $\ket{\psi}_{NM}$ as specified by $U: \cH_X \otimes \cH_N \rightarrow  \cH_{X^\prime}  \otimes \cH_{N^\prime}$ and $V: \cH_Y \otimes \cH_M \rightarrow \cH_{Y'} \otimes \cH_{M'}$.  Let $\hat{\rho} = (U \otimes V)(\sigma \otimes \ket{\psi}_{}\bra{\psi}_{})(U \otimes V)^\dagger$ and ${\rho}$ be the final state after measuring the registers $(X'Y')$ in computational basis~\footnote{We enforce the quantum split-state adversary to return classical registers $X',Y'$ by doing the measurement in computational basis. This is without any loss of generality since the decoding process can start by doing the measurement (this will not affect anything in the absence of quantum split-state adversary) and then decode using $\dec$.}.
        
\end{definition}

\begin{figure}
\centering
\begin{tikzpicture}

\node at (1,4.5) {$S$};
\draw (1.2,4.5) -- (5,4.5);
\draw (5,4.5) -- (15,4.5);
\node at (14.5,4.7) {$S$};

\draw (2,1.5) -- (3,1.5);
\draw (3,-0.5) rectangle (4.5,3.5);
\node at (3.8,1.5) {$\enc$};

\draw  (2,1.5) -- (2,4.5);
\draw [dashed] (4.78,-0.8) -- (4.78,5.5);
\draw [dashed] (8.3,-0.8) -- (8.3,5.5);
\draw [dashed] (10,-0.8) -- (10,5.5);

\node at (6.2,1.6) {$\ket{\psi}_{NM}$};
\node at (4.64,-0.8) {$\sigma$};
\node at (8.1,-0.8) {$\hat{\rho}$};
\node at (10.2,-0.8) {${\rho}$};

\node at (4.6,3) {$X$};
\node at (8,3) {$X'$};
\draw (4.5,2.8) -- (6,2.8);
\draw (7,2.8) -- (8.5,2.8);
\draw (9.5,2.8) -- (11,2.8);
\node at (14.5,2.6) {$S'$};
\draw (14,2.4) -- (15,2.4);

\node at (4.6,0.4) {$Y$};
\node at (8,0.0) {$Y'$};
\draw (4.5,0.2) -- (6,0.2);
\draw (7,0.2) -- (8.5,0.2);
\draw (9.5,0.2) -- (11,0.2);

\draw (6,2) rectangle (7,3);
\node at (6.5,2.5) {$U$};
\draw (6,0) rectangle (7,1);
\node at (6.5,0.5) {$V$};

\node at (6.5,-0.4) {$\mathcal{A}=(U,V,\psi)$};
\draw (5.2,-0.8) rectangle (7.8,3.8);

\draw (5.5,1.5) ellipse (0.2cm and 1cm);
\node at (5.5,2) {$N$};
\draw (5.7,2.2) -- (6,2.2);
\node at (7.4,2.1) {$N'$};
\draw (7,2.2) -- (7.2,2.2);
\node at (5.5,1) {$M$};
\draw (5.7,0.7) -- (6,0.7);
\node at (7.4,0.9) {$M'$};
\draw (7.0,0.7) -- (7.2,0.7);

\draw (9,2.8) circle (0.5);
\node at (9,2.8) {$\mathcal{M}$};
\draw (9,0.2) circle (0.5);
\node at (9,0.2) {$\mathcal{M}$};

\draw (11,-0.5) rectangle (14,3.2);
\node at (12.5,1.5) {$\dec$};

\end{tikzpicture}
\caption{Quantum split-state adversary along with the process.}\label{fig:splitstate1}
\end{figure}
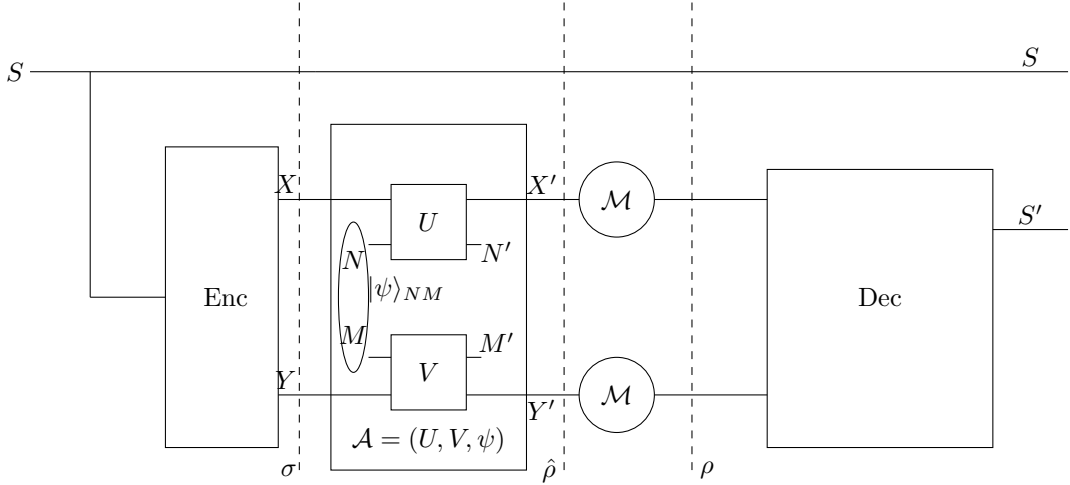

Our work provides the first quantum secure non-malleable code with efficient encoding and decoding procedures for message length $m=n^{\Omega(1)}$, error $\eps=2^{- n^{\Omega(1)}}$ and codeword of size $2n$. When the tampering of the codeword is performed $t$-times, we also provide the first quantum secure one-many non-malleable code with efficient encoding and decoding procedures for $t=n^{\Omega(1)}$, message length $m=n^{\Omega(1)}$, error $\eps=2^{-n^{\Omega(1)}}$ and codeword of size $2n$. Prior to our work, it remained open to provide such quantum secure non-malleable codes even for a single bit message in the  split-state model. We next formally define the quantum secure non-malleable codes in the split-state model.

\begin{definition}[Quantum secure non-malleable codes in the split-state model\label{def:nmcodes_qs}]  $(\enc, \dec)$ is an $(m,n,\eps)$-quantum secure non-malleable code in the split-state model with error $\eps$, if for state $\rho$ and adversary $\mathcal{A} = (U,V, {\psi})$ (as defined in Definition~\ref{def:splitstateadv}), there exists a random variable $\mathcal{D}_{\mathcal{A}}$~\footnote{Distribution depends only on $\mathcal{A}$ and is independent of the original message $S$.} on $\{0,1 \}^m  \cup \lbrace \sm \rbrace$ such that 
 $$\forall s \in \{0,1\}^m: \qquad \Vert S'_s- \cp (\mathcal{D}_{\mathcal{A}} ,s)   \Vert_1  \leq \eps.$$
 Above $S'=\dec(X',Y'), S'_s = (S'|S=s)$ and the function $\cp(x,s)$ is such that $\cp(x,s) =s$ if $x=\mathsf{same}$, otherwise $\cp(x,s) =x$.
\end{definition}

 Our first result is to show that a quantum secure non-malleable code in the split-state model can be constructed using a quantum secure $2$-source non-malleable extractor. We use the $2$-source non-malleable extractor of Boddu, Jain and Kapshikar~\cite{BJK21}. This is analogous to the classical result by Cheraghchi and Guruswami~\cite{CG14b}, however additional novelty over classical arguments is needed. This is to take care of the specific adversary model in which the security of $2$-source non-malleable extractor is shown by~\cite{BJK21} and other additional issues involving quantum information (example purifications of states).

\begin{theorem}[Quantum secure non-malleable codes in the split-state model]\label{thm:introqnmcodesfromnmext} Let $2\nmext : \{0,1\}^{n} \times \{0,1\}^{n} \to \{0,1\}^m$ be an $(n-k,n-k,\eps)$-quantum secure $2$-source non-malleable extractor. There exists an $(m,n,\eps')$-quantum secure non-malleable code in the split-state model with parameter $\eps' = 2^m(4(2^{-k} + \eps) + \eps)$~\footnote{All the quantum secure non-malleable codes in the split-state model we construct are qualitatively same as the non-malleable codes in the split-state model constructed in~\cite{CGL15}.}.
\end{theorem}
Above, $\enc, \dec$ for quantum secure non-malleable code in the split-state model are $2\nmext^{-1},  \newline 2\nmext$ respectively. It can be noted that computation of $Z=2\nmext(X,Y)$ (starting from $(X,Y)$) is efficient (in $n$). This ensures that $\dec$ is efficient. However $\enc$ involves given $Z$, sampling uniformly from the pre-image of $Z$ (under the function $2\nmext$) and it is not apriori clear that this is possible. In the next result, we show that this is indeed possible. This result is analogous to a result due to~\cite{CGL15}, however in our case additional novelty is needed which we explain in the proof overview below. 
\begin{theorem}\label{thm:intromain}
There exists an $(m,n,\eps)$-quantum secure non-malleable code in the split-state model with efficient encoding and decoding procedures for message length $m=n^{\Omega(1)}$, error $\eps=2^{-n^{\Omega(1)}}$ and codeword of size $2n$.
\end{theorem}

Prior to this work, it remained open to provide such construction for quantum secure non-malleable codes, even for a single bit message in the split-state model. 

As an application, we construct the first quantum secure $2$-out-of-$2$ non-malleable secret sharing scheme for message/secret length $m=n^{\Omega(1)}$, error $\eps=2^{-n^{\Omega(1)}}$ and share of size $n$ (see Appendix~\ref{Appendixss}).

We also study the  natural extension when the tampering of the codeword is performed $t$-times (see Appendix~\ref{Appendix1}). Here, the adversary is allowed to tamper 
$$(X,Y) \to \newline (X^1X^2 \ldots X^t, Y^1Y^2 \ldots Y^t)$$ making use of an arbitrary entanglement between two parts $X$ and $Y$. We require, in case of tampering, the original message $S$ to be  independent of $S^{1}\ldots S^t = \dec(X^1,Y^1)\ldots\dec(X^t,Y^t)$.

\begin{theorem}[Quantum secure one-many non-malleable codes in the split-state model]\label{thm:introqnmcodesfromnmextt} Let $t\mhyphen2\nmext : \{0,1\}^{n} \times \{0,1\}^{n} \to \{0,1\}^m$ be a $(t;n-k,n-k,\eps)$-quantum secure $2$-source $t$-non-malleable extractor. There exists a $(t;m,n,\eps')$-quantum secure one-many non-malleable code with parameter $\eps' =2^m(2^{2t}(2^{-k} + \eps) + \eps)$.
\end{theorem}

\begin{theorem}\label{thm:intromain1}
There exists a $(t;m,n,\eps)$-quantum secure one-many non-malleable code in the  split-state model with efficient encoding and decoding procedures for parameters $t=n^{\Omega(1)}$, $m=n^{\Omega(1)}$, error $\eps=2^{-n^{\Omega(1)}}$ and codeword of size $2n$.
\end{theorem}

\subsubsection*{Proof overview}

 Let $\nmextc$ refer to the $2$-source non-malleable extractor from~\cite{CGL15}. Let $XY = U_n \otimes U_n$ ($\otimes$ represents independence). Let $Z=\nmextc(X,Y)$. According to the scheme by~\cite{CG14b},  efficient construction of non-malleable codes requires us to, given any $z$, sample efficiently from the distribution $(XY|Z=z)$. It is not apriori clear that such efficient (reverse) sampling for $\nmextc$ is possible. \cite{CGL15} modified $\nmextc$ to come up with a new $2$-source non-malleable extractor (say $\newnmextc$) and exhibited efficient reverse sampling for $\newnmextc$. A key difference between the constructions of $\nmextc$ and $\newnmextc$ is the seeded extractor that is used in the alternating extraction argument (for both the constructions). $\nmextc$ uses the seeded extractor from~[GUV09] while $\newnmextc$ uses a seeded extractor $\IExt$ constructed by~\cite{CGL15}. Two key properties of $\IExt$ that are crucially used are:
\begin{enumerate}
    \item Let $W$ be the source, $S$ be the seed and $O=\IExt(W,S)$ be the output. For $WS=U_n \otimes U_d$, we have $OS=U_m \otimes U_d$. 
    \item $\IExt$ is a bi-linear function. This implies that for every $(o,s)$, one can sample (exactly) from $(W|OS=(o,s))$.
\end{enumerate}
This allows~\cite{CGL15} exact reverse sampling. That is for any $z$, they are able to efficiently sample from the distribution $(XY|Z=z)$ exactly.

There are a few other modifications required to finally make $\newnmextc$ suitable for efficient reverse sampling. For example, the input sources $X$ and $Y$ are divided into $n^{\Omega(1)}$ different blocks (since there are $n^{\Omega(1)}$ rounds of alternating extraction in the construction of $\nmextc$). This enables to use different blocks (each with almost full min-entropy) as sources to seeded extractors in each round of alternating extraction. This further ensures the linear constraints that are imposed in the alternating extraction are on different variables of input sources, $X$, $Y$ in each round which is crucial for the exact reverse sampling argument of~\cite{CGL15}.  

Let us now consider the quantum setting. Let $\nmextq$ refer to the $2$-source non-malleable extractor from~\cite{BJK21}. Again it is not apriori clear that efficient reverse sampling for $\nmextq$ is possible. Hence we modify the $\nmextq$ from~\cite{BJK21} to construct (say $\newnmextq$) in the full version. We follow the argument of dividing the input sources $X$ and $Y$ into different blocks (as stated in previous paragraph) and make necessary modifications to $\nmextq$. Next, we note the seeded extractor used in alternating extraction of both $\nmextq$, $\newnmextq$ is the Trevison extractor (say $\Trev$) which is quantum secure~\cite{DPVR09}. One can modify $\nmextq$ using a similar modification as that of~\cite{CGL15}, by considering $\IExt$ instead of $\Trev$. However then one would need to first show the quantum security of $\IExt$. This is not known as of now and we leave it for future work. For now we choose to make the arguments work with $\Trev$. We note the two key properties for $\Trev$:
\begin{enumerate}
    \item  For  $WS=U_n \otimes U_d$, we have $OS \approx U_m \otimes U_d$, where  $O=\Trev(W,S)$ ($\approx$ represent close in $\ell_1$ norm).
    \item For every $s$, $\Trev$ is a linear function of $W$. Hence for every $(o,s)$, we can sample efficiently (exactly) from $W|(OS)=(o,s)$. 
\end{enumerate}
Point 1. above is the differentiating property between $\IExt$ and $\Trev$. Hence, unlike~\cite{CGL15}, we cannot do exact reverse sampling and can only do approximate reverse sampling. We therefore have to carefully keep the overall error introduced under control. 

While generating $Z=\newnmextq(X,Y)$, starting from $(X,Y)$, several intermediate random variables (say $(R_1, R_2, ..., R_k)$ in this order) are generated. During the reverse sampling, starting from $Z$, they need to be generated in the reverse order. We call this process {\em backtracking}. Since we have to keep the overall error under control, we need to note and use important Markov-chain structures between the intermediate random variables (see Claim~\ref{fact:markovapprox} and Claim~\ref{claim:extractorsampling}). This is additional technical novelty over~\cite{CGL15}.

\subsubsection*{Organization}

In Section~\ref{sec:prelims}, we describe useful quantum information facts and other preliminaries. It also contains useful lemmas and claims. 
We describe the existential proof of quantum secure non-malleable codes, i.e. Theorem~\ref{thm:introqnmcodesfromnmext} in   Section~\ref{sec:nmcodes}. Section~\ref{sec:sampling} contains the construction of modified $2$-source non-malleable extractor along with proof of Theorem~\ref{thm:intromain}. The $t$-tampered version of non-malleable codes can be found in the Appendix~\ref{Appendix1}. Appendix~\ref{Appendixss} contains a  quantum secure $2$-out-of-$2$ non-malleable secret sharing scheme.

\section{Preliminaries}
\label{sec:prelims}
Let $n,m,d,t$  represent positive integers and $l, k, k_1, k_2, \delta, \gamma, \eps \geq 0$ represent reals.
\subsection*{Quantum information theory} All the logarithms are evaluated to the base $2$. Let $\X, \Y, \Z$ be finite sets (we only consider finite sets in this paper). For a {\em random variable} $X \in \X$, we use $X$ to denote both the random variable and its distribution, whenever it is clear form the context. We use $x \leftarrow X$ to denote $x$ drawn according to $X$. We also use $x \leftarrow \X$ to denote $x$ drawn uniformly from $\X$. For two random variables $X,Y$ we use $X \otimes Y$ to denote independent random variables.

We call random variables $X, Y$, {\em copies} of each other iff $\Pr[X=Y]=1$.   Let $Y^1, Y^2, \ldots, Y^t$ be random variables. We denote the joint random variable  $Y^1 Y^2 \ldots Y^t$ by $Y^{[t]}$.
Similarly for any subset $\mathcal{S} \subseteq [t]$, we use $Y^{\mathcal{S}}$ to denote the joint random variable comprised of all the $Y^s$ such that $s \in \mathcal{S}$. For a random variable $X \in \{0,1 \}^n$ and $0<d_1 \leq d_2\leq n$, let $\pre(X, d_1,d_2)$ represent the bits from $d_1$ to $d_2$ of $X$, i.e. $X^{[d_1,d_2]}$. Let $U_d$ represent the uniform distribution over $\{0,1 \}^d$.  For a random variable $X \in \F_q^n$ for a prime power $q$, we view $X$ as a row vector $(X^1,X^2, \ldots, X^n)$ where each $X^i \in \F_q$.

Consider a finite-dimensional Hilbert space $\cH$ endowed with an inner-product $\langle \cdot, \cdot \rangle$ (we only consider finite-dimensional Hilbert-spaces). A quantum state (or a density matrix or a state) is a positive semi-definite operator on $\cH$ with trace value  equal to $1$. It is called {\em pure} iff its rank is $1$.  Let $\ket{\psi}$ be a unit vector on $\cH$, that is $\langle \psi,\psi \rangle=1$.  With some abuse of notation, we use $\psi$ to represent the state and also the density matrix $\ketbra{\psi}$, associated with $\ket{\psi}$. Given a quantum state $\rho$ on $\cH$, {\em support of $\rho$}, called $\text{supp}(\rho)$ is the subspace of $\cH$ spanned by all eigenvectors of $\rho$ with non-zero eigenvalues.
 
A {\em quantum register} $A$ is associated with some Hilbert space $\cH_A$. Define $\vert A \vert := \log\left(\dim(\cH_A)\right)$. Let $\mathcal{L}(\cH_A)$ represent the set of all linear operators on the Hilbert space $\cH_A$. For operators $O, O'\in \cL(\cH_A)$, the notation $O \leq O'$ represents the L\"{o}wner order, that is, $O'-O$ is a positive semi-definite operator. We denote by $\mathcal{D}(\cH_A)$, the set of all quantum states on the Hilbert space $\cH_A$. State $\rho$ with subscript $A$ indicates $\rho_A \in \mathcal{D}(\cH_A)$. If two registers $A,B$ are associated with the same Hilbert space, we shall represent the relation by $A\equiv B$. For two states $\rho, \sigma$, we let $\rho \equiv \sigma$ represent that they are identical as states (potentially in different registers). Composition of two registers $A$ and $B$, denoted $AB$, is associated with the Hilbert space $\cH_A \otimes \cH_B$.  For two quantum states $\rho\in \mathcal{D}(\cH_A)$ and $\sigma\in \mathcal{D}(\cH_B)$, $\rho\otimes\sigma \in \mathcal{D}(\cH_{AB})$ represents the tensor product ({\em Kronecker} product) of $\rho$ and $\sigma$. The identity operator on $\cH_A$ is denoted $\id_A$. Let $U_A$ denote maximally mixed state in $\cH_A$. Let $\rho_{AB} \in \mathcal{D}(\cH_{AB})$. Define
$$ \rho_{B} \defeq \tr_{A}{\rho_{AB}} \defeq \sum_i (\bra{i} \otimes \id_{B})
\rho_{AB} (\ket{i} \otimes \id_{B}) , $$
where $\{\ket{i}\}_i$ is an orthonormal basis for the Hilbert space $\cH_A$.
The state $\rho_B\in \mathcal{D}(\cH_B)$ is referred to as the marginal state of $\rho_{AB}$ on the register $B$. Unless otherwise stated, a missing register from subscript in a state represents partial trace over that register. Given $\rho_A\in\mathcal{D}(\cH_A)$, a {\em purification} of $\rho_A$ is a pure state $\rho_{AB}\in \mathcal{D}(\cH_{AB})$ such that $\tr_{B}{\rho_{AB}}=\rho_A$. Purification of a quantum state is not unique.
Suppose $A\equiv B$. Given $\{\ket{i}_A\}$ and $\{\ket{i}_B\}$ as orthonormal bases over $\cH_A$ and $\cH_B$ respectively, the \textit{canonical purification} of a quantum state $\rho_A$ is $\ket{\rho_A} \defeq (\rho_A^{\frac{1}{2}}\otimes\id_B)\left(\sum_i\ket{i}_A\ket{i}_B\right)$. 
A quantum {map} $\cE: \mathcal{L}(\cH_A)\rightarrow \mathcal{L}(\cH_B)$ is a completely positive and trace preserving (CPTP) linear map. A {\em Hermitian} operator $H:\cH_A \rightarrow \cH_A$ is such that $H=H^{\dagger}$. A projector $\Pi \in  \mathcal{L}(\cH_A)$ is a Hermitian operator such that $\Pi^2=\Pi$. A {\em unitary} operator $V_A:\cH_A \rightarrow \cH_A$ is such that $V_A^{\dagger}V_A = V_A V_A^{\dagger} = \id_A$. The set of all unitary operators on $\cH_A$ is  denoted by $\mathcal{U}(\cH_A)$. An {\em isometry}  $V:\cH_A \rightarrow \cH_B$ is such that $V^{\dagger}V = \id_A$ and $VV^{\dagger} = \id_B$. A {\em POVM} element is an operator $0 \le M \le \id$. We use the shorthand $\overline{M} \defeq \id - M$, where $\id$ is clear from the context. We use shorthand $M$ to represent $M \otimes \id$, where $\id$ is clear from the context. 

A {\em classical-quantum} (c-q) state $\rho_{XE}$ is of the form \[ \rho_{XE} =  \sum_{x \in \X}  p(x)\ket{x}\bra{x} \otimes \rho^x_E , \] where ${\rho^x_E}$ are states. In a pure state $\rho_{XEA}$ in which $\rho_{XE}$ is c-q, we call $X$ a classical register and identify random variable $X$ with it with $\Pr(X=x) =p(x)$. For an event $\mathcal{S} \subseteq \cX$, define  $$\Pr(\mathcal{S})_\rho \defeq  \sum_{x \in \mathcal{S}} p(x) \quad ; \quad (\rho|S)\defeq \frac{1}{\Pr(\mathcal{S})_\rho} \sum_{x \in \mathcal{S}} p(x)\ket{x}\bra{x} \otimes \rho^x_E.$$  
For a function $Z:\cX \rightarrow \cZ$, define the following extension of $\rho_{XE}$ 
\[ \rho_{ZXE} \defeq  \sum_{x\in \cX}  p(x) \ket{Z(x)}\bra{Z(x)} \otimes \ket{x}\bra{x} \otimes  \rho^{x}_E.\]

We call an isometry $V: \cH_X \otimes \cH_A \rightarrow \cH_X \otimes \cH_B$, {\em safe} on $X$ iff there is a collection of isometries $V_x: \cH_A\rightarrow \cH_B$ such that the following holds.  For all states $\ket{\psi}_{XA} = \sum_x \alpha_x \ket{x}_X \ket{\psi^x}_A$,
$$V  \ket{\psi}_{XA} =  \sum_x \alpha_x \ket{x}_X V_x \ket{\psi^x}_A.$$
All the isometries considered in this paper are safe on classical registers they act upon. For a function $Z:\cX \rightarrow \cZ$, define $\rho_{Z\hat{Z}XEA}$ to be a pure state extension of $\rho_{XEA}$ generated via a safe isometry $V: \cH_X \rightarrow \cH_X \otimes \cH_Z \otimes \cH_{\hat{Z}}$ ($Z$ classical with copy $\hat{Z}$). For a pure state $\rho_{XE}$ and measurement $\mathcal{M}$ in the computational basis on register $X$,  define $\rho_{\hat{X}XE}$ a pure state extension post the measurement  $\mathcal{M}$ of state $\rho_{XE}$ generated via a safe isometry $V: \cH_X \rightarrow \cH_X \otimes \cH_{\hat{X}}$ such that $\rho_{\hat{X}XE} = V \rho V^\dagger$ and $\hat{X}$ a copy of $X$.


\begin{definition}
\label{def:infoquant}    
\begin{enumerate}
\item For $p \geq 1$ and matrix $A$,  let $\| A \|_p$ denote the {\em Schatten} $p$-norm defined as $\| A \|_p  \defeq (\tr(A^\dagger A)^{\frac{p}{2}})^{\frac{1}{p}}.$

\item  For states $\rho,\sigma: \Delta(\rho , \sigma) \defeq \frac{1}{2} \|\rho - \sigma\|_1$. We write $\rho \approx_\eps \sigma$ to denote $\Delta(\rho, \sigma) \le \eps$. 

\item {\bf Fidelity:}  For states $\rho,\sigma: ~\F(\rho,\sigma)\defeq\|\sqrt{\rho}\sqrt{\sigma}\|_1.$ 

\item {\bf Bures metric:}  For states $\rho,\sigma: \Delta_B(\rho,\sigma)\defeq \sqrt{1-\F(\rho,\sigma)}.$ 


\item {\bf Max-divergence (\cite{Datta09}, see also~\cite{JainRS02}):}\label{dmax}  For states $\rho,\sigma$ such that $\supp(\rho) \subset \supp(\sigma)$, $$ \dmax{\rho}{\sigma} \defeq  \min\{ \lambda \in \mathbb{R} :   \rho  \leq 2^{\lambda} \sigma \}.$$ 
\item {\bf Min-entropy and conditional-min-entropy:}  For a state $\rho_{XE}$, the min-entropy of $X$ is defined as,
 $$ \hminone{X}_\rho \defeq - \dmax{\rho_{X}}{\id_X} .$$
 The conditional-min-entropy of $X$, conditioned on $E$, is defined as,
 $$ \hmin{X}{E}_\rho \defeq - \inf_{\sigma_E \in  \mathcal{D}(\cH_{E}) } \dmax{\rho_{XE}}{\id_X \otimes \sigma_E}.$$
  \item {\bf Markov-chain:}\label{cqcmarkov}  A state $\rho_{XEY}$ forms a Markov-chain (denoted $(X-E-Y)_\rho$) iff $\condmutinf{X}{Y}{E}_\rho=0$~\footnote{$\condmutinf{A}{B}{C}_\rho$ represents conditional mutual information between registers $A,B$ given register $C$ in state $\rho$.}.


\end{enumerate}
\end{definition}

For the facts stated below without citation, we refer the reader to standard text books~\cite{NielsenC00,WatrousQI}.
\begin{fact}[Uhlmann's Theorem~\cite{uhlmann76}]
\label{uhlmann}
Let $\rho_A,\sigma_A\in \mathcal{D}(\cH_A)$. Let $\rho_{AB}\in \mathcal{D}(\cH_{AB})$ be a purification of $\rho_A$ and $\sigma_{AC}\in\mathcal{D}(\cH_{AC})$ be a purification of $\sigma_A$. There exists an isometry $V$ (from a subspace of $\cH_C$ to a subspace of $\cH_B$) such that,
 $$ \Delta_B(\ketbra{\theta}_{AB}, \ketbra{\rho}_{AB}) =  \Delta_B(\rho_A,\sigma_A) ,$$
 where $\ket{\theta}_{AB} = (\id_A \otimes V) \ket{\sigma}_{AC}$.
\end{fact}

\begin{fact}[\cite{CLW14}]
	\label{fact102}  
	 Let $\cE :    \mathcal{L} (\cH_M ) \rightarrow   \mathcal{L}(\cH_{M'} )$ be a CPTP map and let $\sigma_{XM'} =(\id \otimes \cE)(\rho_{XM}) $. Then,  $$ \hmin{X}{M'}_\sigma  \geq \hmin{X}{M}_\rho  .$$
Above is equality iff $\cE$ is a CPTP map corresponding to an isometry.
\end{fact}

\begin{fact}[Stinespring isometry extension~\cite{WatrousQI}]\label{fact:stinespring}
		 Let $\Phi :    \mathcal{L} (\cH_X ) \rightarrow   \mathcal{L}(\cH_Y )$ be a CPTP map. Then there exists an isometry $V :  \cH_{X} \rightarrow   \cH_{Y} \otimes \cH_{Z}$ (Stinespring isometry extension of $\Phi$) such that $\Phi(\rho_X)= \tr_{Z}(V \rho_X V^\dagger)$ for every state $\rho_X$.
\end{fact} 

\begin{fact}[\cite{FvdG06}]
\label{fidelty_trace}
Let $\rho,\sigma$ be states. Then,
\[  1-\F(\rho,\sigma) \leq \Delta(\rho , \sigma) \leq \sqrt{ 1-\F^2(\rho,\sigma)} \quad ; \quad \Delta_B^2(\rho,\sigma) \leq \Delta(\rho , \sigma) \leq  \sqrt{2}\Delta_B(\rho,\sigma).  \]

\end{fact}
\begin{fact}[Data-processing]
\label{fact:data}
Let $\rho, \sigma$ be states and $\cE$ be a CPTP map. Then 
\begin{itemize}
    \item $ \Delta ( \cE(\rho)  , \cE(\sigma))  \le \Delta (\rho  , \sigma).$    
     \item $ \Delta_B ( \cE(\rho)  , \cE(\sigma))  \le \Delta_B (\rho  , \sigma).$    
    \item  $\dmax{ \cE(\rho) }{ \cE(\sigma) }  \le \dmax{\rho}{ \sigma} .$    
\end{itemize}
The inequalities above are equalities in case $\Phi$ is a CPTP map corresponding to an isometry.
\end{fact}

\begin{fact}[]
\label{traceavg}
Let $\rho_{XE},\sigma_{XE}$ be c-q states. Then,
\begin{itemize}
    \item $ \| \rho_{XE}-\sigma_{XE} \|_1 \geq   \E_{x \leftarrow \rho_X } \| \rho^x_{E}-\sigma^x_{E} \|_1. $
     \item $ \Delta_B( \rho_{XE},\sigma_{XE} ) \geq   \E_{x \leftarrow \rho_X } \Delta_B( \rho^x_{E}, \sigma^x_{E} ). $
\end{itemize}
The above inequalities are equalities iff $\rho_X = \sigma_X$.
\end{fact}

\begin{fact}\label{fact:traceconvex} Let $\rho, \sigma$  be states such that $\rho = \sum_{x} p_x \rho^x$,  $\sigma = \sum_{x} p_x \sigma^x$, $\{\rho^x, \sigma^x\}_x$ are states and $\sum_x p_x =1$. Then, 
 $\Delta(\rho, \sigma) \leq \sum_x p_x \Delta(\rho^x,\sigma^x)$.

\end{fact}

\begin{fact}
\label{measuredmax}
Let $\rho_{AB} \in \mathcal{D}(\cH_A \otimes \cH_B)$ be a state and $M \in \cL(\cH_B)$ such that $M^\dagger M \leq \id_B$. Let $\hat{\rho}_{AB}= \frac{M \rho_{AB} M^\dagger}{\tr{M \rho_{AB} M^\dagger}}$. Then, 
$$\dmax{\hat{\rho}_{A}}{\rho_{A}} \leq \log \left(\frac{1}{\tr{M \rho_{AB} M^\dagger}}\right).$$
\end{fact}

\begin{fact}\label{fact:tensortrace}
For random variables $A,B$ such that $A=A_1\otimes A_2$, $B=B_1\otimes B_2$, we have 
$$ \| A-B \|_1 = \| A_1-B_1\|_1+\| A_2-B_2\|_1. $$

\end{fact}

\begin{fact}\label{fact:usefultracedistance} For random variables $AB$, $\tilde{A}\tilde{B}$, we have 
$$\Vert \tilde{A}\tilde{B}-AB\Vert_1 \leq \Vert \tilde{B}-B\Vert_1+ \mathbb{E}_{b \leftarrow \tilde{B}}
        \Vert \tilde{A}\vert(\tilde{B}=b)
        )-A\vert(B=b)\Vert_1.$$
\end{fact}

\begin{fact}[Folklore]\label{fact:samplingeffifeintly}
Let $m,n$ be positive integers such that $m\leq n$. Let $A$ be any $m \times n$ matrix over the Field $\mathbb{F}$. For any string $o \in \mathbb{F}^m$, let $\mathsf{S}_o \defeq \{x \in \mathbb{F}^n   : Ax^\dagger =o^\dagger \}.$ There exists an efficient algorithm that runs in time polynomial in $(m,n,\vert \mathbb{F} \vert)$ and outputs sample $x \leftarrow \mathsf{S}_o$. 
\end{fact}

\begin{fact}[Corollary 5.2 in~\cite{CGL15}]\label{fact:samp}
  For any constant $\delta  \in ( 0,1)$, there exist constants $\alpha,\beta$  such that $3\beta \leq \alpha \leq 1/14$ and  for all positive integers $\nu, r, t$, with $r \geq \nu^{\alpha}$ and $t=\cO(\nu^{\beta})$ the following holds.  
  
  There exists a polynomial time computable function 
  $\samp : \{0,1 \}^r \to [\nu]^{t}$, such that for any set $\mathcal{S} \subset [\nu]$ of size $\delta \nu$,
  $$ \Pr( \vert \samp(U_r) \cap \mathcal{S}  \vert \geq 1 ) \geq 1-2^{- \Omega(\nu^{\alpha})}.$$
\end{fact}

\begin{definition} Let $M=2^m$. The inner-product function, $\IP^{n}_{M}: \mathbb{F}_{M}^{n} \times \mathbb{F}_{M}^n \rightarrow \mathbb{F}_{M}$ is defined as follows:\[\IP^{n}_{M}(x,y)=\sum_{i=1}^{n} x_i y_i,\]
where the operations are over the Field $\mathbb{F}_{M}.$
\end{definition}

\subsection*{Extractors and non-malleable codes}
Throughout the paper we use extractor to mean seeded extractor unless stated otherwise. 
\begin{definition}[Quantum secure extractor]
\label{qseeded}
	An $(n,d,m)$-extractor $\Ext : \{0,1\}^n \times \{0,1\}^d \to \{0,1\}^m$  is said to be $(k,\eps)$-quantum secure if for every state $\rho_{XES}$, such that $\Hmin(X|E)_\rho \geq k$ and $\rho_{XES} = \rho_{XE} \otimes U_d$, we have 
	$$  \| \rho_{\Ext(X,S)E} - U_m \otimes \rho_{E} \|_1 \leq \eps.$$
	In addition, the extractor is called strong if $$  \| \rho_{\Ext(X,S)SE} - U_m \otimes U_d \otimes \rho_{E} \|_1 \leq \eps .$$
	$S$ is referred to as the {\em seed} for the extractor.
	\end{definition}

\begin{fact}[\cite{DPVR09,CV16}]
    \label{fact:extractor} There exists an explicit $(2m,\eps)$-quantum secure strong $(n,d,m)$-extractor $\Ext : \{ 0,1\}^n \times  \{ 0,1\}^d \to  \{ 0,1\}^m$ for parameters  $d = \cO( \log^2(n/\eps) \log m )$. Moreover the extractor $\Ext$ is linear extractor, i.e. for every fixed seed, the output of the extractor is a linear function of the input source.
\end{fact}

\begin{definition}[$l\mhyphen\qmas$~\cite{ABJO21}]\label{qmadv}	Let $\tau_{X\hat{X}}$, $\tau_{Y\hat{Y}}$ be the canonical purifications of independent and uniform sources $X, Y$ respectively. Let $\tau_{NM}$ be a pure state. Let 
$$ \theta_{X\hat{X}NMY\hat{Y}}= \tau_{X\hat{X}} \otimes \tau_{NM} \otimes \tau_{Y\hat{Y}}.$$
Let $U : \cH_{X} \otimes \cH_{N} \rightarrow   \cH_{X} \otimes \cH_{N'} \otimes \cH_{A}$ and $V : \cH_Y \otimes \cH_{M} \rightarrow   \cH_{Y} \otimes \cH_{M'} \otimes \cH_{B}$ be isometries
such that registers $A, B$ are single qubit registers. Let $$\rho_{X\hat{X}AN'M'BY\hat{Y}} = (U \otimes V)\theta_{X\hat{X}NMY\hat{Y}}(U \otimes V)^\dagger,$$
and 
 $$l = \log\left( \frac{1}{ \Pr(A=1, B=1)_{\rho}}  \right) \quad ; \quad \sigma_{X\hat{X}N'M'Y\hat{Y}} =(\rho_{X\hat{X}AN'M'BY\hat{Y}} \vert A=1,B=1).$$
 We call $\sigma_{X\hat{X}N'M'Y\hat{Y}}$ an $l\mhyphen\qmas$ .
\end{definition}
\begin{fact}[$\IP$ security against $l$-$\qmas$~\cite{ABJO21}] \label{l-qma-needed-fact} Let $n=\frac{n_1}{m}$ and $n_1-l \geq 2 \log\left(\frac{1}{\eps}\right)+m$. Let $\sigma_{X \hat{X} N' Y \hat{Y} M'}$ be an $l\mhyphen\qmas$ with $\vert X \vert = \vert Y\vert = n_1$. Then
\[\Vert \sigma_{\IP^n_{2^m}(X,Y)XN'} - U_{m} \otimes \sigma_{XN'}  \Vert_1 \leq \eps \quad ; \quad \Vert \sigma_{\IP^n_{2^m}(X,Y)YM'} - U_{m} \otimes \sigma_{YM'}  \Vert_1 \leq \eps.\]
\end{fact}
\begin{fact}[Theorem 8 in~\cite{CGL15}]\label{classicalip}
Let $X,Y$ be independent sources on $\mathbb{F}^{n}_{2^m}$ with min-entropy $k_1,k_2$ respectively. Then
\[\Vert \IP^n_{2^m}(X,Y)Y -U_m \otimes Y \Vert_1 \leq \eps \quad ; \quad  \Vert \IP^n_{2^m}(X,Y)X -U_m \otimes X  \Vert_1 \leq \eps,\] where $\eps=2^{\frac{-(k_1+k_2-m(n+1))}{2}}.$
\end{fact}

\begin{definition}[$(k_1,k_2)\mhyphen\qmas$~\cite{BJK21}]\label{qmadvk1k2}
We call a pure state $\sigma_{X\hat{X}NMY\hat{Y}}$, with $(XY)$ classical and $(\hat{X}\hat{Y})$ copy of $(XY)$,  a  $(k_1,k_2)\mhyphen\qmas$ iff 
\[ \hmin{X}{MY\hat{Y}}_\sigma \geq k_1 \quad ; \quad \hmin{Y}{NX\hat{X}}_\sigma \geq k_2.\]
\end{definition}
\begin{definition}[$(k_1,k_2)\mhyphen\nmas$~\cite{BJK21}]\label{def:2source-qnmadversarydef}
     Let $\sigma_{X\hat{X}NMY\hat{Y}}$ be a $(k_1,k_2)\mhyphen\qmas$. Let $U: \cH_X \otimes \cH_N \rightarrow \cH_X \otimes \cH_{X^\prime} \otimes  \cH_{\hat{X}'} \otimes \cH_{N^\prime}$ and $V: \cH_Y \otimes \cH_M \rightarrow \cH_Y \otimes \cH_{Y'} \otimes  \cH_{\hat{Y}'} \otimes \cH_{M'}$ be isometries  such that for $\rho = (U \otimes V)\sigma(U \otimes V)^\dagger,$ we have $(X'Y')$ classical (with copy $\hat{X}'\hat{Y}'$) and, 
       $$\Pr(Y \ne Y^\prime)_\rho =1 \quad or \quad \Pr(X \ne X^\prime)_\rho =1.$$ 
       We call state $\rho$ a $(k_1,k_2)\mhyphen\nmas$.
\end{definition}
\begin{definition}[Quantum secure $2$-source non-malleable extractor~\cite{BJK21}]\label{def:2nme}
		An $(n,n,m)$-non-malleable extractor $2\nmext : \{0,1\}^{n} \times \{0,1\}^{n} \to \{0,1\}^m$ is $(k_1,k_2,\eps)$-secure against $\nma$ if for every $(k_1,k_2)\mhyphen\nmas$ $\rho$ (chosen by the adversary $\nma$),
	$$  \Vert \rho_{ 2\nmext(X,Y)2\nmext(X^\prime,Y^\prime) Y  Y^\prime M^\prime} - U_m \otimes \rho_{ 2\nmext(X^\prime ,Y^\prime) Y  Y^\prime M^\prime} \Vert_1 \leq \eps, $$ and 
	$$  \Vert \rho_{ 2\nmext(X,Y)2\nmext(X^\prime,Y^\prime) X  X^\prime N^\prime} - U_m \otimes \rho_{ 2\nmext(X^\prime ,Y^\prime) X X^\prime N^\prime} \Vert_1 \leq \eps.$$
\end{definition}


\subsection*{Error correcting codes} 
 
 \begin{definition}\label{def:ecc}
     Let $\Sigma$ be a finite set. A mapping $\ecc: \Sigma^k \to \Sigma^n$ is called an error correcting code with relative distance $\gamma$ if for any $x,y \in \Sigma^k$ such that $x \ne y,$ the Hamming distance between $\ecc(x)$ and $\ecc(y)$ is at least $\gamma n.$ The rate of the code denoted by $\delta$, is defined as $\delta \defeq \frac{k}{n}$. The alphabet size of the code is the number of elements in $\Sigma.$
 \end{definition}
 \begin{fact}[MDS Codes]\label{fact:ecc}
         Let $q$ be a prime power.
         For every positive integer $k$, there exists a large enough $n$ such that there exists an efficiently computable linear error correcting code $\ecc: \F^{k}_q \to \F^{ n}_q $ with rate $\frac{k}{n}$ and relative distance $ \frac{n-k+1}{n}$. 
         Such codes are known as maximum distance separable (MDS) codes.
         Reed-Solomon codes is a typical example of an MDS code family.
 \end{fact}

\subsection*{Other useful facts,  claims and lemmas\label{sec:claims}}
\begin{fact}[\cite{BJK21}]\label{fact:minentropydecrease}
Let $\rho_{ABC} \in \mathcal{D}(\cH_A \otimes \cH_B \otimes \cH_C)$ be a state and $M \in \cL(\cH_C)$ such that $M^\dagger M \leq \id_C$. Let $\hat{\rho}_{ABC}= \frac{M \rho_{ABC} M^\dagger}{\tr{M \rho_{ABC} M^\dagger}}$. Then, 
\[ \hmin{A}{B}_{\hat{\rho}} \geq   \hmin{A}{B}_{\rho} - \log \left(\frac{1}{\tr{M \rho_{ABC} M^\dagger}}\right). \]
\end{fact}

\begin{fact}[\cite{BJK21}]\label{fact:prefixminentropyfact}
     Let $\rho_{XE} \in \mathcal{D}(\cH_X \otimes \cH_E)$ be a c-q state such that $\vert X \vert =n$ and $\hmin{X}{E}_\rho \geq n-k.$ Let $\tilde{X} = \pre(X,d_1,d_2)$ for some positive integer $k \leq d_2-d_1$ and $1 \leq d_1 < d_2 \leq n$. Then, $$\hmin{\tilde{X}}{E}_\rho \geq d_2-d_1+1-k.$$
\end{fact}

\begin{fact}[\cite{BJK21}] \label{lemma:nearby_rho_prime_prime} Let $\rho_{X \hat{X} N Y \hat{Y} M}$ be a $(k_1,k_2)\mhyphen\qmas$ such that $\vert X \vert = \vert \hat{X} \vert= \vert Y \vert= \vert \hat{Y} \vert =n$. There exists an $l\mhyphen\qmas$, $\rho^{(1)}$ such that,
\[ \Delta_B(\rho^{(1)}, \rho) \leq {3}\eps \quad and \quad l \leq 2n- k_1 - k_2+ 4 +  4\log \left( \frac{1}{\eps} \right). \]
\end{fact}

\begin{fact}[Quantum secure $2$-source non-malleable extractor~\cite{BJK21}]\label{fact:thm:2nmext}
Let $k = \cO(n^{1/4})$ and $\eps =2^{-n^{\Omega(1)}}$. There exists an efficient $2$-source non-malleable extractor $2\nmext :\{0,1 \}^n \times \{0,1 \}^n \to \{0,1 \}^{n/4}$ that is  $(n-k,n-k, \cO(\eps))$-quantum secure.
\end{fact}

\begin{fact}[Alternating extraction~\cite{BJK21}]\label{lem:2}
Let $\theta_{XASB}$ be a pure state with $(XS)$ classical, $\vert X \vert =n, \vert S \vert =d$ and
\[ \hmin{X}{SB}_\theta \geq k \quad ; \quad \Delta_B( \theta_{X A S} , \theta_{X A} \otimes U_d ) \leq \eps^{\prime}. \]
Let $T \defeq \Ext(X,S)$ where $\Ext$ is a $(k,\eps)$-quantum secure strong $(n,d,m)$-extractor. Then, 
\[\Delta_B( \theta_{T B} , U_m \otimes \theta_{B} ) \leq  2\eps' + \sqrt{\eps}.\]
\end{fact}
\begin{fact}[Min-entropy loss under classical interactive communication~\cite{BJK21}]\label{lem:minentropy}
Let $\rho_{XNM}$ be a pure state where Alice holds registers $(XN)$ and Bob holds register $M$, such that register $X$ is classical and
\[ \hmin{X}{M}_\rho \geq k. \]
Let Alice and Bob proceed for $t$-rounds, where in each round Alice generates a classical register $R_i$ and sends it to Bob, followed by Bob generating a classical register $S_i$ and sending it to Alice. Alice applies an isometry $V^{i}: \cH_X \otimes \cH_{N_{i-1}} \rightarrow \cH_X \otimes \cH_{N'_{i-1}} \otimes \cH_{R_{i}}$ (in round $i$) to generate $R_{i}$. Let 
 $\theta^i_{XN_iM_i}$ be the state at the end of round-$i$, where Alice holds registers $XN_i$ and Bob holds register $M_i$. Then,
 \[ \hmin{X}{M_t}_{\theta^t} \geq k-\sum_{j=1}^{t} \vert R_j\vert .\]
\end{fact}
\begin{claim}
\label{factcor:2nmextext}
 Let $\sigma_{X\hat{X}NY\hat{Y}M}$ be a $(k_1,k_2)\mhyphen \qmas$ with $\vert X \vert=n$ and $\vert Y \vert=n$. Let $2\nmext:\{0,1\}^n \times\{0,1\}^{n} \to \lbrace 0,1\rbrace^{m}$ be an efficient  $(k_1,k_2, \eps)$-quantum secure $2$-source non-malleable extractor. Let $S=2\nmext(X,Y)$. Then,
$$ \| \sigma_{SYM} - U_{m} \otimes \sigma_{YM} \|_1 \leq \eps.$$
\end{claim}
\begin{proof}
Let $U: \cH_X \rightarrow \cH_X \otimes \cH_{X'} \otimes  \cH_{\hat{X}'}$, $V: \cH_Y \rightarrow \cH_Y \otimes \cH_{Y'} \otimes  \cH_{\hat{Y}'}$ be (safe) isometries such that for $\rho = (U \otimes V) \sigma (U \otimes V)^\dagger$, we have $X'Y'$ classical (with copies $\hat{X}'\hat{Y}'$ respectively) and either  $\Pr(X \ne X')_\rho =1$ or $\Pr(Y \ne Y')_\rho =1.$~\footnote{It is easily seen that such isometries exists.} Notice the state $\rho$ is a $(k_1,k_2)\mhyphen\nmas$. Since $2\nmext$ is a  $(k_1,k_2, \eps)$-quantum secure $2$-source  non-malleable extractor (see Definition~\ref{def:2nme}), we have $$ \| \rho_{SS'YY'M} - U_{m} \otimes \rho_{S'YY'M} \|_1 \leq \eps
.$$ 
Using Fact~\ref{fact:data}, we further get $$ \| \rho_{SYM} - U_{m} \otimes \rho_{YM} \|_1 \leq \eps.$$ 
The desired now follows by noting $\sigma_{XNMY} =  \rho_{XNMY}.$
\end{proof}

\begin{claim}\label{fact:markovapprox}
Let random variables $ABC$, $\tilde{A}\tilde{B}\tilde{C}$ be such that
\[A\leftrightarrow B \leftrightarrow C \quad ;\quad  \tilde{A}\leftrightarrow \tilde{B} \leftrightarrow \tilde{C} \quad ; \quad\Vert AB -\tilde{A}\tilde{B}  \Vert_1 \leq  \eps_1 \quad ;\quad \Vert BC - \tilde{B}\tilde{C} \Vert_1 \leq  \eps_2.\]Then, 
\[\Vert    ABC-\tilde{A}\tilde{B}\tilde{C} \Vert \leq 2\eps_1+ \eps_2.\]
\end{claim}
\begin{proof}
Since $\Vert AB -\tilde{A}\tilde{B}  \Vert_1 \leq  \eps_1 $, using Fact~\ref{fact:data}, we have
\begin{equation}\label{eq:factmarkov9888}
    \Vert B -\tilde{B}  \Vert_1 \leq  \eps_1.
\end{equation}
Since $A \leftrightarrow B \leftrightarrow C $, $\tilde{A} \leftrightarrow \tilde{B} \leftrightarrow \tilde{C}$, we have
\begin{equation}\label{eq:factmarkov9999}
   (AC)\vert(B=b)=A\vert(B=b)\otimes C\vert(B=b) \quad ;\quad  (\tilde{A}\tilde{C})\vert(\tilde{B}=b) =\tilde{A}\vert(\tilde{B}=b) \otimes \tilde{C}\vert(\tilde{B}=b).
\end{equation}Consider,\begin{align*}
        \Vert  \tilde{A}\tilde{B}\tilde{C}-ABC\Vert_1 & \leq \Vert \tilde{B}-B\Vert_1+ \mathbb{E}_{b \leftarrow \tilde{B}}
        \Vert (\tilde{A}\tilde{C})\vert(\tilde{B}=b)
        )-(AC)\vert(B=b)\Vert_1& \mbox{(Fact~\ref{fact:usefultracedistance})}\\
        & = \Vert \tilde{B}-B\Vert_1+ \mathbb{E}_{b \leftarrow \tilde{B}}
        \Vert \tilde{A}\vert(\tilde{B}=b) \otimes \tilde{C}\vert(\tilde{B}=b)
        -A\vert(B=b) \otimes C\vert(B=b)\Vert_1& \mbox{(Eq.~\eqref{eq:factmarkov9999})}\\
         & = \Vert \tilde{B}-B\Vert_1+ \mathbb{E}_{b \leftarrow \tilde{B}} \left(\Vert \tilde{A}\vert(\tilde{B}=b)- A\vert(B=b)\Vert_1+ \Vert \tilde{C}\vert(\tilde{B}=b)- C\vert(B=b)\Vert_1 \right)& \mbox{(Fact~\ref{fact:tensortrace})}\\
         & =\Vert \tilde{B}-B\Vert_1+ \mathbb{E}_{b \leftarrow \tilde{B}}
        \left(\Vert \tilde{A}\vert(\tilde{B}=b)- A\vert(B=b)\Vert_1 \right)+\mathbb{E}_{b \leftarrow \tilde{B}}\left(\Vert \tilde{C}\vert(\tilde{B}=b)- C\vert(B=b)\Vert_1 \right)& \mbox{}\\
        & \leq\Vert \tilde{B}-B\Vert_1+\Vert \tilde{A}\tilde{B}-AB\Vert_1+\Vert \tilde{B}\tilde{C}-BC\Vert_1& \mbox{(Fact~\ref{traceavg})}\\
        & \leq \eps_1+\Vert \tilde{A}\tilde{B}-AB\Vert_1+\Vert \tilde{B}\tilde{C}-BC\Vert_1& \mbox{(Eq.~\eqref{eq:factmarkov9888})}\\
        & \leq 2\eps_1+\eps_2.
        \end{align*}
This completes the proof.
\end{proof}

Since the construction of the quantum secure  non-malleable extractor is composed of alternating extraction using $\Ext$ from Fact~\ref{fact:extractor}, we first state a claim about the invertibility of the $\Ext$ given the output (close to the desired).
\begin{claim}\label{claim:extractorsampling}
  Let $\Ext : \{ 0,1\}^n \times  \{ 0,1\}^d \to  \{ 0,1\}^m$ be an explicit $(2m,\eps)$-quantum secure strong extractor from Fact~\ref{fact:extractor} or $\IP$ from Fact~\ref{classicalip}~\footnote{Inputs are of same size in this case, i.e. $d=n$.} with error $\eps$~\footnote{We assume there is enough min-entropy in the sources for the extractors to work.}. Let $X,H,O,\tilde{O}$ be random variables such that,
  \[ \Vert XH -U_n \otimes U_d \Vert_1 \leq  \eps'   \quad ; \quad O =\Ext(X,H) \quad ; \quad \Vert  \tilde{O} -U_m\Vert_1 \leq \eps''.  \]
  Given samples from $(\tilde{o}, \tilde{h}) \leftarrow \tilde{O}\tilde{H} = \tilde{O} \otimes U_d$, we can sample from $\tilde{X} \vert (\tilde{O}\tilde{H} =\tilde{o}\tilde{h})$ (which is same as  sampling $\tilde{x}  \leftarrow\{ x :  \Ext(x,\tilde{h}) =\tilde{o} \}$) in time polynomial in $(m,n)$ such that 
  \[\Vert \tilde{O}\tilde{H}\tilde{X} -OHX \Vert_1 \leq \eps +\eps'+\eps'' \quad ; \quad \tilde{O}=\Ext(\tilde{X},\tilde{H}). \]
\end{claim}
\begin{proof}
Let $ \hat{X}\hat{H} =U_n \otimes U_d $ and $\hat{O} = \Ext(\hat{X},\hat{H})$ be the output of the extractor. Since $ \Vert XH -  \hat{X}\hat{H} \Vert_1  \leq \eps',$ using Fact~\ref{fact:data} we have
\begin{equation}\label{eq:claimtrevext1}
    \Vert OXH - \hat{O}\hat{X}\hat{H} \Vert_1  \leq \eps'.
\end{equation}
Also, since $\hat{O} = \Ext(\hat{X},\hat{H})$ is the output of the strong extractor, we have 
\begin{equation*}
      \Vert \hat{O}\hat{H} - U_m \otimes U_d \Vert_1 \leq \eps.
\end{equation*}
Since $\tilde{O}\tilde{H} =\tilde{O}\otimes U_d$ and $\Vert \tilde{O} -U_m \Vert_1 \leq \eps''$, we have $\Vert \tilde{O}\tilde{H} - U_m \otimes U_d \Vert_1 \leq \eps''.$ Using triangle inequality, we have 
\begin{equation}\label{eq:claimtrevext2}
      \Vert \hat{O}\hat{H} - \tilde{O}\tilde{H} \Vert_1 \leq \eps+\eps''.
\end{equation}

We now proceed by noting that the extractor is linear. In other words, for every seed $H=h$, the output of the extractor $O=o$ is a linear function of the input $X=x$. For a fixed output $o$ of the extractor and seed $h$, we have a matrix $A_h$ of size $m \times n$ such that $A_h x^\dagger =o^\dagger$. Note for any fixing of the seed $h$ and output $o$, the size of
the set $\{ x  :  \Ext(x,h) =o \}$ is $2^{n- \rk(A_h) }$ and sampling $x$ uniformly from the set can be done efficiently from Fact~\ref{fact:samplingeffifeintly}.

Given samples $(\tilde{o}, \tilde{h}) \leftarrow \tilde{O}\tilde{H} =\tilde{O} \otimes U_d$, we use the following sampling strategy:
\begin{itemize}
    \item Sample $ \tilde{x} \leftarrow \tilde{X} \vert (\tilde{O}\tilde{H}=\tilde{o}\tilde{h})$ which is same as  sampling $\tilde{x} \leftarrow \{ x :  \Ext(x,\tilde{h}) =\tilde{o} \}$.
\end{itemize}
Note for any fixing of the seed $\tilde{h}$ and output $\tilde{o}$ of the extractor, $\hat{X} \vert (\hat{O}\hat{H}=\tilde{o}\tilde{h})=\tilde{X} \vert (\tilde{O}\tilde{H}=\tilde{o}\tilde{h})$ since $\hat{X}\hat{H}=U_n \otimes U_d$. Thus, from Fact~\ref{fact:data} and Eq.~\eqref{eq:claimtrevext2}, we have\begin{equation}\label{eq:claimtrevext3}
     \Vert \hat{O}\hat{X}\hat{H} -\tilde{O}\tilde{X}\tilde{H} \Vert_1 \leq \eps+\eps''.
\end{equation}Using Eqs.~\eqref{eq:claimtrevext1},~\eqref{eq:claimtrevext3} along with  triangle inequality, we have the desired.
\end{proof}
\begin{corollary}\label{corr:claim4}
  Let $\Ext : \{ 0,1\}^n \times  \{ 0,1\}^d \to  \{ 0,1\}^m$ be an explicit $(2m,\eps)$-quantum secure strong extractor from Fact~\ref{fact:extractor} or $\IP$ from Fact~\ref{classicalip}~\footnote{Inputs are of same size in this case, i.e. $d=n$.} with error $\eps$. Let 
  \[ \Vert XH -U_n \otimes U_d \Vert_1 \leq  \eps'   \quad ; \quad O =\Ext(X,H).  \]Then,
  $\Vert O -U_m  \Vert_1 \leq  \eps+ \eps'.$
\end{corollary}
\begin{proof}

Let $ \hat{X}\hat{H} =U_n \otimes U_d $ and $\hat{O} = \Ext(\hat{X},\hat{H})$ be the output of the extractor. Since $ \Vert XH -  \hat{X}\hat{H} \Vert_1  \leq \eps',$ using Fact~\ref{fact:data} we have\begin{equation*}
    \Vert OXH - \hat{O}\hat{X}\hat{H} \Vert_1  \leq \eps'.
\end{equation*}
Using Fact~\ref{fact:data} again, we further  have\begin{equation}\label{eq:claimdfvftrevext1}
    \Vert O - \hat{O}\Vert_1  \leq \eps'.
\end{equation}Also, since $\hat{O} = \Ext(\hat{X},\hat{H})$ is the output of the  extractor, we have 
\begin{equation*}
      \Vert \hat{O} - U_m  \Vert_1 \leq \eps.
\end{equation*}
Using Eq.~\eqref{eq:claimdfvftrevext1} along with triangle inequality, we have the desired. 
\end{proof}

\begin{lemma}\label{lem:reedsolomon}Let $\ecc : \F^{k}_q \to \F^{n}_q$ be an $(n,k,n-k+1)$ Reed-Solomon code from Fact~\ref{fact:ecc} for $q \geq n+1.$ Let random variable $M \in \F^{k}_q$ be uniformly distributed over $\F^{k}_q$. Let $C= \ecc(M)$ and $t$ be any positive integer such that $t < k$. Let $\mathcal{S}$ be a subset of $[n]$ such that $\vert \mathcal{S}\vert =t$ and  $\mathcal{Q}$ be a subset of $[k]$ such that $\vert \mathcal{Q}\vert =j \leq k-t.$ Then, for every fixed string $c$ in $\F^{t}_q$ and $C^{\mathcal{S}} = c$~\footnote{$C^{\mathcal{S}}$ corresponds to codeword corresponding to columns  $\mathcal{S}$ of codeword $C$.}, the distribution $M^c \defeq M \vert (C^{\mathcal{S}} = c)$ is such that $(M^c)^{\mathcal{Q}} =U_{j \log q}$. Further more, for any fixed string 
$l$ in $\F^{j}_q$, we can efficiently (in time polynomial in $(k,q)$) sample from the distribution $(M^c)^{[k] \setminus \mathcal{Q}} \vert ((M^c)^{\mathcal{Q}} =l ).$

\end{lemma}
\begin{proof}
The generator matrix for $\ecc$ is given by
$$G = \begin{pmatrix} 1 & \alpha_{1} & \ldots & \alpha^{{k-1}}_{1} \\ 1 & \alpha_{2} & \ldots & \alpha^{{k-1}}_{2} \\ \vdots & \vdots & \ddots & \vdots \\ 1 & \alpha_{n} & \ldots & \alpha^{{k-1}}_{n} \end{pmatrix},$$ 
where $\alpha_1, \alpha_2, \ldots, \alpha_n$ are distinct non-zero elements of $\mathbb{F}_{q}$ (this is possible since $q \geq n+1$). Let $\mathcal{S}= \{ s_1, s_2, \ldots, s_t \}$ and $$G_{\mathcal{S}} = \begin{pmatrix} 1 & \alpha_{s_1} & \ldots & \alpha^{{k-1}}_{s_1} \\ 1 & \alpha_{s_2} & \ldots & \alpha^{{k-1}}_{s_2} \\ \vdots & \vdots & \ddots & \vdots \\ 1 & \alpha_{s_{t}} & \ldots & \alpha^{{k-1}}_{s_{t}} \end{pmatrix}.$$ Note we have   $G_{\mathcal{S}} M^\dagger = ({C^{\mathcal{S}}})^\dagger$. By fixing $C^{\mathcal{S}}=c$, we have imposed the following linear constraints as given by
$G_{\mathcal{S}} M^c{^\dagger} = c^\dagger$. Note $G_{\mathcal{S}}$ is a Vandermonde matrix for any fixed subset $\mathcal{S} \subset [n]$, $ \vert \mathcal{S} \vert =t $ and $t <k$. Thus, any $t \times t$ submatrix of $G_{\mathcal{S}}$ has full rank. Note $k-j \geq t.$ Note $M^c{}$ is exactly the distribution,
$ m \leftarrow\{ m \in \F_q^k : G_{\mathcal{S}} m^\dagger = c^\dagger\}$.  Let $\mathcal{P} = [k] \setminus \mathcal{Q}$ and $\mathcal{P}= \{ p_1, p_2, \ldots, p_{k-j}\}$ with elements in the set $\mathcal{P}$ in any fixed order. Equivalent way to define $M^c{}$ is the distribution,
$m \leftarrow\{ m \in \F_q^k : \tilde{G} m^\dagger = \tilde{c}^\dagger\}$, such that $\tilde{G}$ is $t \times k$ matrix, the submatrix of  $\tilde{G}$ corresponding to columns given by $\mathcal{P}'= \{p_1,p_2, \ldots, p_t\}$ is exactly $\id_{t \times t}$ (since any $t \times t$ submatrix of $G_{\mathcal{S}}$ has full rank). Note one can get $(\tilde{G},\tilde{c})$ from  $({G}_{\mathcal{S}},{c})$ using standard Gaussian elimination procedure (in time polynomial in $(k,q)$). Thus, sampling $m=(m_1, m_2, \ldots, m_t)$ from the distribution $M^c{}$ can be achieved as follows:
\begin{itemize}
    \item Sample for every $i \in \mathcal{Q}$, $m_i$ uniformly and independently from $\F_q$.
     \item Sample for every $i \in \mathcal{P}\setminus \mathcal{P}'$, $m_i$ uniformly and independently from $\F_q$.
    \item For every $i \in \mathcal{P}'$, set $m_i \in \F_q$ such that it satisfies the linear constraints $\tilde{G}_{i}m^\dagger =\tilde{c}^\dagger$.~\footnote{Note the variables in $\mathcal{Q}$ that appear in this linear constraint are already sampled before.}
\end{itemize}
Thus, $(M^c)^{\mathcal{Q}} =U_{j \log q}$. Further more, for any fixed string 
$l$ in $\F^{j}_q$, we can efficiently (in time polynomial in $(k,q)$) sample from the distribution $(M^c)^{[k] \setminus \mathcal{Q}} \vert ((M^c)^{\mathcal{Q}} =l ).$ 
\end{proof}

 \section{A quantum secure non-malleable code in the split-state model\label{sec:nmcodes}}

\begin{theorem}[Quantum secure non-malleable code in the split-state model]\label{lem:qnmcodesfromnmext} Let $2\nmext : \{0,1\}^{n} \times \{0,1\}^{n} \to \{0,1\}^m$ be an $(n-k,n-k,\eps)$-quantum secure $2 \mhyphen$source non-malleable extractor. There exists an $(m,n,\eps')$-quantum secure non-malleable code with parameter $\eps' = 2^m(4(2^{-k} + \eps) + \eps)$.
\end{theorem}
\begin{proof}
Let $\sigma_S = U_m$. Given a $2\nmext : \{0,1\}^{n} \times \{0,1\}^{n} \to \{0,1\}^m$, let $(\enc,\dec)$ be defined as follows: 
\begin{itemize}
    \item For any fixed message $S=s \in  \{0, 1\}^m$, the encoder, $\enc$ outputs a uniformly random string from the set  $2\nmext^{-1}(s) \subset  \{ 0,1\}^{2n}$. 
    \item $\dec(x,y) = 2\nmext (x,y)$ for every $(x,y) \in  \{0, 1\}^{2n}$.
\end{itemize}
Note for every $s \in \{0,1 \}^m$, $\Pr(\dec(\enc(s)))=1$.

We first show that $(\enc, \dec)$ is a quantum secure non-malleable code in the split-state model (see Definition~\ref{def:nmcodes_qs}~and~Figure~\ref{fig:splitstate1}). Note $$ \sigma_{S\enc(S)}=\sigma_{SXY} =  \sum_{s \in \{0,1 \}^m} \left( \frac{1}{ 2^{m}}\ketbra{s } \otimes  \left( \frac{1}{\vert 2\nmext^{-1}(s)  \vert} \sum_{(x,y) : 2\nmext(x,y)=s} \ketbra{xy} \right)\right) $$ after encoding a uniform message $S$. Note $(X,Y)=\enc(S)$.

Let $\theta_{X\hat{X}X_1 Y\hat{Y}Y_1 } = \theta_{X\hat{X}X_1 }  \otimes \theta_{Y\hat{Y}Y_1 }$ be a pure state such that $\theta_X = \theta_Y =U_n$, ($X_1, \hat{X}$) are copies of $X$, ($Y_1, \hat{Y}$) are copies of $Y$ respectively. Note $\theta_{XY} =U_{n} \otimes U_n$. Let  $S_1=2\nmext(X_1,Y_1)$, $S=2\nmext(X,Y)$. For the state $\theta$ with the following assignment (terms on the left are from Definition~\ref{qmadvk1k2} and on the right are from here),
\[(X,\hat{X},N,M,Y,\hat{Y}) \leftarrow (X_1,\hat{X},X,Y,Y_1,\hat{Y}),\]one can note $\theta$ is an $(n,n)\mhyphen \qmas$. Using Claim~\ref{factcor:2nmextext}~\footnote{$(n,n)\mhyphen \qmas$ is also an $(n-k,n-k)\mhyphen \qmas$.} along with Fact~\ref{fact:data}, we have 
$$ \Vert \theta_{2\nmext(X_1 , Y_1)} - U_m \Vert_1  \leq  \eps  .$$Thus,
\begin{equation}\label{eq:proof}
     \Vert\theta_{S_1XY} -  \sigma_{SXY} \Vert_1 \leq  \Vert \theta_{S_1} -   U_m \Vert_1  \leq \eps.
\end{equation}First inequality follows from Fact~\ref{fact:data} and noting 
$$\theta_{XY} \vert (S_1=s)=\sigma_{XY} \vert (S=s) = \left( \frac{1}{\vert 2\nmext^{-1}(s)  \vert} \sum_{(x,y) : 2\nmext(x,y)=s} \ketbra{xy} \right).$$

Let $\mathcal{A}=(U,V,{\psi})$ be the quantum split-state adversary from Definition~\ref{def:splitstateadv}. Note ${\psi}_{NM}$ is an entangled pure state, $U: \cH_{X} \otimes \cH_N \rightarrow \cH_{X^\prime} \otimes \cH_{N^\prime}$ and $V: \cH_{Y} \otimes \cH_M \rightarrow  \cH_{Y'} \otimes \cH_{M'}$ are isometries without any loss of generality.

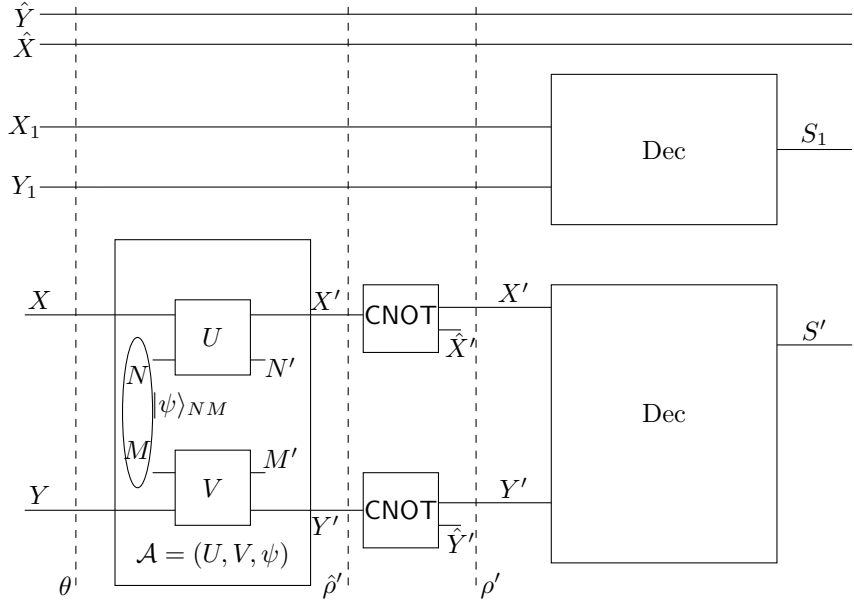
\begin{figure}
\centering
\begin{tikzpicture}


\draw (4.2,6.4) -- (15,6.4);
\node at (4,6.4) {$\hat{X}$};

\draw (4.2,6.8) -- (15,6.8);
\node at (4,6.8) {$\hat{Y}$};



\draw (4.2,5.3) -- (11,5.3);
\node at (4,5.3) {$X_1$};

\node at (4,4.5) {$Y_1$};
\draw (4.2,4.5) -- (5,4.5);
\draw (5,4.5) -- (11,4.5);
\node at (14.5,5.2) {$S_1$};
\draw (14,5) -- (15,5);

\draw (11,4) rectangle (14,6);
\node at (12.5,5) {$\dec$};


\draw [dashed] (4.68,-0.8) -- (4.68,7);
\draw [dashed] (8.3,-0.8) -- (8.3,7);
\draw [dashed] (10,-0.8) -- (10,7);

\node at (6.2,1.6) {$\ket{\psi}_{NM}$};
\node at (4.54,-0.8) {$\theta$};
\node at (8.1,-0.8) {$\hat{\rho}'$};
\node at (10.2,-0.8) {${\rho}'$};

\node at (4.2,3) {$X$};
\node at (8,3) {$X'$};
\node at (10.5,3.1) {$X'$};
\node at (9.8,2.4) {$\hat{X}'$};
\draw (4,2.8) -- (6,2.8);
\draw (7,2.8) -- (8.5,2.8);
\draw (9.5,2.9) -- (11,2.9);
\draw (9.5,2.6) -- (9.8,2.6);
\node at (14.5,2.6) {$S'$};
\draw (14,2.4) -- (15,2.4);

\node at (4.2,0.4) {$Y$};
\node at (8,0.0) {$Y'$};
\draw (4,0.2) -- (6,0.2);
\draw (7,0.2) -- (8.5,0.2);
\draw (9.5,0.3) -- (11,0.3);
\draw (9.5,0) -- (9.8,0);
\node at (10.5,0.5) {$Y'$};
\node at (9.8,-0.2) {$\hat{Y}'$};

\draw (6,2) rectangle (7,3);
\node at (6.5,2.5) {$U$};
\draw (6,0) rectangle (7,1);
\node at (6.5,0.5) {$V$};

\node at (6.5,-0.4) {$\mathcal{A}=(U,V,\psi)$};
\draw (5.2,-0.8) rectangle (7.8,3.8);

\draw (5.5,1.5) ellipse (0.2cm and 1cm);
\node at (5.5,2) {$N$};
\draw (5.7,2.2) -- (6,2.2);
\node at (7.4,2.1) {$N'$};
\draw (7,2.2) -- (7.2,2.2);
\node at (5.5,1) {$M$};
\draw (5.7,0.7) -- (6,0.7);
\node at (7.4,0.9) {$M'$};
\draw (7.0,0.7) -- (7.2,0.7);

\draw (8.5,2.3) rectangle (9.5,3.2);
\node at (9,2.8) {$\mathsf{CNOT}$};

\draw (8.5,-0.3) rectangle (9.5,0.7);
\node at (9,0.2) {$\mathsf{CNOT}$};

\draw (11,-0.5) rectangle (14,3.2);
\node at (12.5,1.5) {$\dec$};

\end{tikzpicture}
\caption{Analysis of a quantum secure non-malleable code in the split-state model}\label{fig:splitstate2}
\end{figure}

In the analysis, we consider a pure state ${\rho}'$ which is generated from $\theta_{X\hat{X}X_1Y\hat{Y}Y_1 } = \theta_{X\hat{X}X_1 }  \otimes \theta_{Y\hat{Y}Y_1}$, in the following way (see Figure~\ref{fig:splitstate2}): 
\begin{itemize}
\item  Let $\hat{\rho}' = (U \otimes V)(\theta \otimes \ket{\psi}_{}\bra{\psi}_{})(U \otimes V)^\dagger$ be the state after the action of quantum split-state adversary.
\item Let ${\rho}'$ be the pure state extension after measuring the registers $(X'Y')$ in the computational basis in $\hat{\rho}'$. Note the measurement in the  computational basis of registers $(X',Y')$ corresponds to applying $\mathsf{CNOT}$~\footnote{By $\mathsf{CNOT}$ on $n$ qubit register, we mean applying $\mathsf{NOT}$ gate on different ancilia each time conditioned on each qubit independently. This operation amounts to performing measurement in the computational basis.} to modify $(X',Y')$  $ \to (X'\hat{X}',Y'\hat{Y}')$ such that $\hat{X}',\hat{Y}'$ are copies of $X',Y'$ respectively.


\end{itemize}

Let binary variables $C,D$ (with copies $\hat{C},\hat{D}$) be such that $C=1$ indicates $X_1 \ne X'$ and $D=1$ indicates $Y_1 \ne Y'$ (in state ${\rho}'$). 

Let  $S'=2\nmext(X',Y')$. Since $ \Vert\theta_{S_1XY} -  \sigma_{SXY} \Vert_1  \leq \eps$ from Eq.~\eqref{eq:proof}, using Fact~\ref{fact:data} we have
$ \Vert \rho'_{S_1S'} -\rho_{SS'} \Vert_1 \leq \eps.$ We will show

\begin{equation}\label{eq:imp2}
    \Vert \rho'_{S_1S'}- Z  \cp (\mathcal{D}_{\mathcal{A}} ,Z)   \Vert_1  \leq 4(2^{-k}+\eps),
\end{equation}
for $Z=U_m$ and distribution  $\mathcal{D}_{\mathcal{A}}$ that depends only on $\mathcal{A}$. We get that $$\Vert {\rho}_{SS'}- Z  \cp (\mathcal{D}_{\mathcal{A}} ,Z)   \Vert_1  \leq 4(2^{-k}+\eps)+\eps,$$ from  $\Vert \rho'_{S_1S'} -\rho_{SS'} \Vert_1 \leq \eps $ and the triangle inequality, which implies the desired (using Fact~\ref{traceavg}), 
$$\forall s \in \{0,1\}^m: \qquad \Vert S'_s- \cp (\mathcal{D}_{\mathcal{A}} ,s)   \Vert_1  \leq 2^m(4(2^{-k}+\eps)+\eps).$$

We now proceed to prove Eq.~\eqref{eq:imp2}. For $C=c \in \{ 0,1\}, D=d \in \{ 0,1\}$, denote $\rho'^{c,d} = \rho'\vert ((C,D)=(c,d))$.

\begin{claim}\label{claim:987}For every $c,d \in \{ 0,1 \}$ except $(c,d)=(0,0)$, we have 
$$\Pr((C,D)=(c,d))_{\rho'} \Vert   \rho'^{c,d}_{S_1S' } - Z \otimes \rho'^{c,d}_{S'} \Vert_1 \leq 2^{-k} + \eps .$$
For $(c,d)=(0,0)$, we have 
$$\Pr((C,D)=(c,d))_{\rho'} \Vert  \rho'^{c,d}_{S_1S' }  - ZZ \Vert_1 \leq 2^{-k} + \eps ,$$ where $Z=U_m$.
\end{claim}
\begin{proof}

Fix $c,d \in \{0,1\}$. Suppose $\Pr((C,D)=(c,d))_{\rho'}  \leq 2^{-k}$, then we are done. Thus we assume otherwise. Note in state $\rho'$, we have
\[ \rho'_{X_1Y_1Y'\hat{Y}'\hat{Y}M'D \hat{D}} =U_n  \otimes \rho'_{Y_1Y'\hat{Y}'\hat{Y} M'D \hat{D}} \quad ; \quad \rho'_{Y_1X_1X'\hat{X}'\hat{X}N'C \hat{C}} =U_n  \otimes \rho'_{X_1X'\hat{X}'\hat{X}N'C \hat{C}}.\] Thus,
\[\hmin{X_1}{Y_1Y'\hat{Y}'\hat{Y}M'D \hat{D}}_{\rho'}=n \quad ;\quad \hmin{Y_1}{X_1X'\hat{X}'\hat{X}N'C \hat{C}}_{\rho'}=n. \]
Using Fact~\ref{fact102}, we have 
\[  \hmin{X_1}{Y_1Y'\hat{Y'}\hat{Y}M'}_{\rho'} = n \quad ; \quad  \hmin{Y_1}{X_1X'\hat{X'}\hat{X}N'}_{\rho'} = n. \]We use Fact~\ref{fact:minentropydecrease},  with the following assignment of registers (below the registers on the left are from Fact~\ref{fact:minentropydecrease} and the registers on the right are the registers in this proof),
 $$(\rho_A, \rho_B, \rho_{C}) \leftarrow (\rho'_{X_1} ,  \rho'_{Y_1Y'\hat{Y'}\hat{Y}M'}, \rho'_{CD}).$$From Fact~\ref{fact:minentropydecrease}, we get that
 \[ \hmin{X_1}{Y_1Y'\hat{Y'}\hat{Y}M'}_{\rho'^{c,d}} \geq n-k.\] 
 Similarly $\hmin{Y_1}{X_1X'\hat{X'}\hat{X}N'}_{\rho'^{c,d}} \geq n-k.$

Let $(c,d)=(0,0)$. Note $\rho'^{0,0}$ is an  $(n-k,n-k) \mhyphen\qmas$. Using Claim~\ref{factcor:2nmextext} along with Fact~\ref{fact:data}, we have 
$\Vert   \rho'^{0,0}_{ 2\nmext( X_1,Y_1) } - Z  \Vert_1 \leq \eps.$ We also have $X_1=X'$ and $Y_1=Y'$ in $\rho'^{0,0}$. Thus Fact~\ref{fact:data} will imply 
$$\Vert \rho'^{0,0}_{S_1S' } - ZZ \Vert_1 \leq \eps.$$

Let $(c,d) \ne (0,0)$. Note $\rho'^{c,d}$ is an $(n-k,n-k)\mhyphen\nmas$ since either $\Pr(X_1 \neq X')=1$ or $\Pr(Y_1 \neq Y')=1$ in $\rho'^{c,d}$. Since, $2\nmext : \{0,1\}^{n} \times \{0,1\}^{n} \to \{0,1\}^m$ is an $(n-k,n-k,\eps)$-quantum secure $2 \mhyphen$source non-malleable extractor (see Definition~\ref{def:2nme}), using Fact~\ref{fact:data} we have 
$$\Vert   \rho'^{c,d}_{S_1S' } - Z \otimes \rho'^{c,d}_{S'   } \Vert_1 \leq \eps.$$
This completes the proof. 
\end{proof}
For every $c,d \in \{ 0,1 \}$ except $(c,d)=(0,0)$, let  $\mathcal{D}^{c,d}_{\mathcal{A}} = \rho'^{c,d}_{S'} $ and for $(c,d)=(0,0)$, let $\mathcal{D}^{0,0}_{\mathcal{A}}$ be the distribution that is deterministically equal to $\sm$. Let $\mathcal{D}_{\mathcal{A}} = \sum_{c,d \in \{ 0,1 \}} \Pr((C,D)=(c,d))_{\rho'} \mathcal{D}^{c,d}_{\mathcal{A}}$. Note for every $c,d \in \{ 0,1 \}$, the value $\Pr((C,D)=(c,d))_{\rho'}$ depends only on $\mathcal{A}$. We have,
\begin{equation}\label{eq:newproof1}
   \rho'_{S_1S'} = \sum_{c,d \in \{ 0,1\}}\Pr((C,D)=(c,d))_{\rho'} \rho'^{c,d}_{S_1S'},
\end{equation}and \begin{equation}\label{eq:newproof2}
     Z  \cp (\mathcal{D}_{\mathcal{A}} ,Z)= \sum_{c,d \in \{ 0,1\}}\Pr((C,D)=(c,d))_{\rho'} Z\cp(\mathcal{D}^{c,d}_{\mathcal{A}} ,Z) .
\end{equation}
Consider,
\begin{align*}
   & \Vert \rho'^{}_{S_1S' }  - Z  \cp (\mathcal{D}_{\mathcal{A}} ,Z)   \Vert_1 \\ 
    & =\Vert \sum_{c,d \in \{ 0,1\}}\Pr((C,D)=(c,d))_{\rho'} \rho'^{c,d}_{S_1S' } - Z  \cp (\mathcal{D}_{\mathcal{A}},Z)   \Vert_1 & \mbox{(Eq.~\eqref{eq:newproof1})}\\ 
   &\leq
    \sum_{c,d \in \{ 0,1\}}\Pr((C,D)=(c,d))_{\rho'} \Vert \rho'^{c,d}_{S_1S' } - Z  \cp (\mathcal{D}^{c,d}_{\mathcal{A}} ,Z)  \Vert_1 & \mbox{(Eq.~\eqref{eq:newproof2} and Fact~\ref{fact:traceconvex})} \\
   & \leq (2^{-k} + \eps)4. & \mbox{(Claim~\ref{claim:987})}
\end{align*} 
This completes the proof.

\end{proof}

 \section{Efficient quantum secure non-malleable codes\label{sec:sampling}}

\subsection{Modified non-malleable extractor\label{subsec:modnmext}}
\noindent These parameters hold throughout this section. 
 \subsection*{Parameters}
Let $\delta, \delta_1, \delta_2 >0$ be small enough constants such that $\delta_1 < \delta_2$. Let  $n,n_1,n_2,n_3,n_4,n_5,n_6,n_7,n_x,n_y,a,s,b,h$ be positive integers and
 $ \eps', \eps > 0$ such that: 
  \[  n_1 = n^{\delta_2} \quad ; \quad  n_2 =n-3n_1  \quad ; \quad q= 2^{\log \left(n+1\right) }  \quad ; \quad \eps= 2^{- \cO(n^{\delta_1})} \quad ;\]
 
\[ n_3 = \frac{n_1}{10} \quad ; \quad n_4 = \frac{n_2}{\log (n+1)} \quad ; \quad n_5 = n^{\delta_2/3} \quad ;\quad a=6n_1+2 \cO( n_5) \log (n+1) = \mathcal{O}(n_1) \quad ;     \]

\[ n_6 = 3n_1^3 \quad ; \quad n_7 = n-3n_1-n_6  \quad ; \quad n_x = \frac{n_7}{12a} \quad ; \quad n_y = \frac{n_7}{12a} \quad ; \quad  2^{\cO(a)}\sqrt{\eps'} = \eps \quad ;  \]

\[s = \cO\left(\log^2\left(\frac{n}{\eps'}\right)\log n \right)  \quad  ; \quad b = \cO\left( \log^2\left(\frac{n}{\eps'}\right) \log n \right) \quad ; \quad h = 10s \quad \]
\begin{itemize}\label{sec:extparameters_2nm}
    \item $\IP_1$ be $\IP^{3n_1/n_3}_{2^{n_3}}$,
    \item $\Ext_1$ be $(2b, \eps')$-quantum secure $(n_y,s,b)$-extractor,
    \item $\Ext_2$ be $(2s, \eps')$-quantum secure $(h,b,s)$-extractor,
    \item $\Ext_3$ be $(4h, \eps')$-quantum secure $(n_x,b,2h)$-extractor,
    \item $\Ext_4$ be $(n_y/4, \eps^2)$-quantum secure $(4n_y,2h,n_y/8)$-extractor,
    \item $\IP_2$ be $\IP^{3n_1^3/2h}_{2^{2h}}$,
    \item  $\Ext_6$ be $(\frac{n_x}{2}, \eps^2)$-quantum secure $(4n_x,n_y/8,n_x/4)$-extractor. 
 \end{itemize}
 We first describe a short overview of the modifications required in the construction of non-malleable extractor for the efficient encoding of quantum secure non-malleable code in the split-state model. We modify the construction of $2\nmext$ from~\cite{BJK21}, using ideas from~\cite{CGL15} to construct a $\mathsf{new}$-$2\nmext : \{0,1\}^{n} \times \{0,1\}^{n} \to \{0,1\}^m$ that is $(n-n_1,n-n_1,\eps)$-secure against $\nma$ for parameters $n_1=n^{\Omega(1)}$, $m= n^{1-\Omega(1)}$ and $\eps = 2^{-n^{\Omega(1)}}$.
 We divide the sources $X$ and $Y$ into $n^{\Omega(1)}$ blocks each of size  $n^{1-\Omega(1)}$.  The idea now is to use new blocks of $X$ and $Y$ for each round of alternating extraction in the construction of  non-malleable extractor. This enables the linear constraints that are imposed in the alternating extraction are on different variables of input sources, $X,Y$. Also, since $X$ and $Y$ each have almost full min-entropy, we have block sources, where each block has almost full min-entropy using Fact~\ref{fact:prefixminentropyfact}. This allows us to generate appropriate intermediate seed random variables (approximately uniform) using alternating extraction. 
 \subsection*{Definition of modified non-malleable extractor}
 
 Let $\ecc : \F^{n_4}_q \to \F^{n}_q$ be an $(n,n_4,n-n_4+1)$ Reed-Solomon code from Fact~\ref{fact:ecc}. Let $\samp : \{0,1 \}^r \to [n]^{t_1}$ be the sampler function from Fact~\ref{fact:samp} where $t_1=\cO(n_5)$ and $r \geq n_3$. We identify the output of $\samp$ as $t_1$ samples from the set $[n]$.  By $\ecc(Y)_{\samp(I)}$, we mean the $\samp(I)$ entries of codeword $\ecc(Y)$.

\begin{algorithm}
\caption{: $\mathsf{new}$-$2\nmext: \lbrace 0,1 \rbrace ^n\times \lbrace 0,1 \rbrace^n   \rightarrow \lbrace 0,1 \rbrace^{n_x/4}$}\label{alg:2nmExtnew}
\begin{algorithmic}
\State{}

\noindent \textbf{ Input:}  $X, Y$\\


\begin{enumerate}
    \item Advice generator: \[X_1=\pre(X,1,3n_1) \quad  ; \quad Y_1 = \pre(Y,1,3n_1) \quad ;
    \quad R= \IP_1(X_1,Y_1) \quad ; \quad  \]
    \[X_2=\pre(X,3n_1+1,n) \quad  ; \quad Y_2 = \pre(Y,3n_1+1,n) \quad ;\]
    \[ \quad G= X_1 \circ \bar{X}_2 \circ Y_1 \circ \bar{Y}_2 = X_1  \circ \ecc(X_2)_{\samp(R)} \circ Y_1 \circ \ecc(Y_2)_{\samp(R)} \]

    \item $X_3=\pre(X,3n_1 +1,3n_1+n_6) \quad ; \quad Y_3=\pre(Y,3n_1 +1,3n_1+n_6)\quad; \quad Z^1=\IP_2(X_3,Y_3)$
     \item $X_4=\pre(X,3n_1+n_6+1,n) \quad ; \quad Y_4=\pre(Y,3n_1+n_6+1,n)\quad$
    \item Correlation breaker with advice:\quad $F=2\advcb(Y_4,X_4,Z^1, G)$
    \item $X^{a+1}=\pre(X_4,4n_xa+1,4n_xa+4n_x)$
    \item $S=\Ext_6(X^{a+1},F)$
\end{enumerate}

 \noindent \textbf{ Output:} $S$ 
\end{algorithmic}
\end{algorithm}

\begin{algorithm}
\caption{: $2\advcb: \lbrace 0,1 \rbrace^{n_7} \times \lbrace 0,1 \rbrace^{n_7} \times \lbrace 0,1\rbrace^{2h} \times \lbrace 0,1 \rbrace^a \rightarrow \lbrace 0,1 \rbrace^{\frac{n_y}{8}}$}\label{alg:2AdvCB}
\begin{algorithmic}
\State{}

\noindent \textbf{ Input: }$Y_4, X_4, Z^1, G$

\begin{enumerate}
    \item For $i=1,2,\ldots,a:$ 
    
     \item  \hspace{1cm}  $Y^{i}=\pre(Y_4,4n_y(i-1)+1,4n_y(i))$ 
     \item  \hspace{1cm}  $X^{i}=\pre(X_4,4n_x(i-1)+1,4n_x(i))$ 
    \item  \hspace{1cm}Flip-flop: \quad $Z^{i+1}=2\ff(Y^i,X^i,Z^{i},G_i)$ 
    
     \item For $i=a+1,a+2,\ldots,3a:$ 
    \item  \hspace{1cm}  $Y^{i}=\pre(Y_4,4n_y(i-1)+1,4n_y(i))$ 
     \item  \hspace{1cm}  $X^{i}=\pre(X_4,4n_x(i-1)+1,4n_x(i))$ 
     \item $F=\Ext_4(Y^{a+1},Z^{a+1})$

\end{enumerate}

 \noindent \textbf{ Output:} $F$ 
\end{algorithmic}
\end{algorithm} 


\begin{algorithm}
\caption{: $2\ff : \lbrace 0,1 \rbrace^{4n_y} \times \lbrace 0,1 \rbrace^{4n_x} \times \lbrace 0,1 \rbrace^{2h} \times \lbrace 0,1 \rbrace \rightarrow \lbrace 0,1\rbrace^{2h}$}\label{alg:2FF}
\begin{algorithmic}
\State{}

\noindent\textbf{Input:}  $Y^i, X^i,  Z^i,  G_i$  \\ 
 $X^i_{j} = \pre(X^i,n_x(j-1)+1,n_x(j))$ for every $j \in \{1,2,3,4 \}$ \\
 $Y^i_{j} = \pre(Y^i,n_y(j-1)+1,n_y(j))$ for every $j \in \{1,2,3,4 \}$ \\
 $Z^i_1 = \pre(Z^i,1,h)$  and $Z^i_2 = \pre(Z^i,h+1,2h)$\\ \\
\begin{enumerate}
     
    \item $Z^i_s=$Prefix$(Z^i_1,s)$, $A^i= \Ext_1 (Y^i_1,Z^i_s),\ C^i= \Ext_2(Z^i_2,A^i),\ B^i= \Ext_{1}(Y^i_2,C^i)$
    \item If $G_i=0$ then $\overline{Z}^i= \Ext_3(X^i_1,A^i)$ and if $G_i=1$ then $\overline{Z}^i= \Ext_3(X^i_2,B^i)$
    \item $\overline{Z}^i_1 = \pre(\overline{Z}^i,1,h)$  and $\overline{Z}^i_2 = \pre(\overline{Z}^i,h+1,2h)$.
     \item $\overline{Z}^i_s=$Prefix$(\overline{Z}^i_1,s)$, $\overline{A}^i= \Ext_1 ({Y^i_3},\overline{Z}^i_s),\ \overline{C}^i= \Ext_2(\overline{Z}^i_2,\overline{A}^i),\ \overline{B}^i= \Ext_{1}({Y^i_4},\overline{C}^i)$
    \item If $G_i=0$, then $O=\Ext_3(X^i_3,\overline{B}^i)$ and if $G_i=1$, then $O=\Ext_3(X^i_4,\overline{A}^i)$
\end{enumerate} \noindent \textbf{Output:} $O$
\end{algorithmic}
\end{algorithm}

\newpage 
\begin{theorem}[Security of $\mathsf{new}$-$2\nmext$] \label{thm:2nmext}

Let $\rho_{X\hat{X}X^\prime \hat{X}' NYY'\hat{Y}\hat{Y}'M}$ be an $(n-n_1,n-n_1)\mhyphen\nmas$ with $\vert X \vert = \vert Y \vert = n$. Then,
\[ \Vert \rho_{ S S^{\prime} Y  Y^\prime M} - U_{n_x/4} \otimes \rho_{ S^\prime Y  Y^\prime M} \Vert_1 \leq \cO( \eps) ,\]
where $S=\mathsf{new}$-$2\nmext(X,Y)$ and $S'=\mathsf{new}$-$2\nmext(X',Y')$.
\end{theorem}
\begin{proof}
The proof proceeds in similar lines to the proof of Theorem $6$ in~\cite{BJK21} using Fact~\ref{lem:2} for alternating extraction argument, Fact~\ref{lem:minentropy} for bounding the min-entropy required in alternating extraction, Facts~\ref{l-qma-needed-fact},~\ref{lemma:nearby_rho_prime_prime} for the security of inner-product function in $(k_1,k_2) \mhyphen \qmas$ framework, Fact~\ref{fidelty_trace} for relation between $\Delta_B, \Delta$ and we do not repeat it. 
\end{proof}
\subsection{Efficiently sampling from the preimage of $\mathsf{new}$-$2\nmext$}
Recall that we showed existence of a quantum secure non-malleable code where encoding scheme was based on inverting $2\nmext$, a quantum secure $2$-source non-malleable extractor. In particular, for any fixed message $S=s$, the encoder, $\enc$ outputs a uniformly random string from the set  $2\nmext^{-1}(s)$. The decoder is the function $2\nmext$ itself. We call this as the \emph{encoding and decoding} based on $2\nmext$. We now state the main result of this paper.
\begin{theorem}[Main Theorem]\label{thm:main}
Let $\mathsf{new}$-$2\nmext: \lbrace 0,1 \rbrace ^n\times \lbrace 0,1 \rbrace^n   \rightarrow \lbrace 0,1 \rbrace^{m}$ be the quantum secure $2$-source non-malleable extractor from Algorithm~\ref{alg:2nmExtnew}, where  $m=n_x/4$. Let $(\enc,\dec)$ be the encoding and decoding based on $\mathsf{new}$-$2\nmext$. Let $\hat{S}\enc(\hat{S})=\hat{S}\hat{X}\hat{Y}$ for a uniform message $\hat{S}=U_{m}$. There exists an efficient algorithm that can sample from a distribution $\tilde{S}\tilde{X}\tilde{Y}$ such that $\Vert \tilde{S}\tilde{X}\tilde{Y} -\hat{S}\hat{X}\hat{Y} \Vert_1 \leq \cO(\eps)$ and $\tilde{S}= U_{m}$. 
\end{theorem}
\begin{proof}Consider $XY=U_n \otimes U_n$. Let $S =\mathsf{new} \mhyphen2\nmext(X,Y)$. From~Eq.~\eqref{eq:proof} in the proof of  Theorem~\ref{lem:qnmcodesfromnmext} (after noting $\hat{S}\hat{X}\hat{Y} \equiv (SXY)_\sigma$ in Eq.~\eqref{eq:proof} and $SXY \equiv (SXY)_\theta$ in Eq.~\eqref{eq:proof}), we have 
\begin{equation}\label{eq:1234}
       \Vert{SXY} -  \hat{S}\hat{X}\hat{Y}  \Vert_1  \leq \cO(\eps).
\end{equation}From Claim~\ref{claim:final}, we have an efficient algorithm that can sample from a distribution $\tilde{S}\tilde{X}\tilde{Y}$ such that 
\[  \Vert{SXY} -  \tilde{S}\tilde{X}\tilde{Y}  \Vert_1  \leq \cO(\eps) \quad ; \quad \tilde{S} =U_m. \]
From Eq.~\eqref{eq:1234} and using triangle inequality, we have 
 \[  \Vert \hat{S}\hat{X}\hat{Y}  - \tilde{S}\tilde{X}\tilde{Y}  \Vert_1  \leq \cO(\eps) \quad ; \quad \tilde{S} =U_m, \] 
which completes the proof.  
\end{proof}
We have the following corollary. 
\begin{corollary} Let $0<\delta_3 < \delta_1$ and $m'=n^{\delta_3}$ be an integer. Let $\mathsf{new}$-$2\nmext: \lbrace 0,1 \rbrace ^n\times \lbrace 0,1 \rbrace^n   \rightarrow \lbrace 0,1 \rbrace^{m}$ be the quantum secure $2$-source non-malleable extractor from Algorithm~\ref{alg:2nmExtnew}, where  $m=n_x/4$. Let $2\nmext: \lbrace 0,1 \rbrace ^n\times \lbrace 0,1 \rbrace^n   \rightarrow \lbrace 0,1 \rbrace^{m'}$ be such that $2\nmext(X,Y)$ is same as $\mathsf{new}$-$2\nmext (X,Y)$ truncated to first $m'$ bits.  Let $(\enc,\dec)$ be the encoding and decoding based on $2\nmext$. Let $\hat{S}\enc(\hat{S})=\hat{S}\hat{X}\hat{Y}$ for a uniform message $\hat{S}=U_{m'}$. There exists an efficient algorithm that can sample from a distribution $\tilde{S}\tilde{X}\tilde{Y}$ such that $\tilde{S}= U_{m'}$ and for every $s \in \{0,1 \}^{m'}$, we have $\Vert (\tilde{X}\tilde{Y}) \vert (\tilde{S}=s) -(\hat{X}\hat{Y}) \vert (\hat{S}=s) \Vert_1 \leq \cO(2^{m'}\eps) \leq 2^{-n^{\Omega(1)}}$. 
\end{corollary}

\begin{claim}\label{claim:final}
Let $2\nmext: \lbrace 0,1 \rbrace ^n\times \lbrace 0,1 \rbrace^n   \rightarrow \lbrace 0,1 \rbrace^{n_x/4}$ be the $\mathsf{new} \mhyphen 2\nmext$ from Algorithm~\ref{alg:2nmExtnew}. Let ${XY}=U_n \otimes U_n$, $S =2\nmext(X,Y)$. We can sample efficiently (in time polynomial in $n$) from a distribution $(\tilde{X}\tilde{Y}\tilde{S})$ such that 
\[\Vert XYS -\tilde{X}\tilde{Y}\tilde{S} \Vert_1 \leq {\cO(\eps)} \quad ;\quad \tilde{S} =U_{n_x/4}.\]

\end{claim}
\begin{proof} 
We consider the intermediate random variables while generating $S =\mathsf{new}\mhyphen2\nmext(X,Y)$ from ${X,Y}$. We use the same intermediate random variables but with tilde notation to denote the random variables we sample (in the reverse order). We note that the argument in this proof is entirely classical. Let
$$X_1,Y_1,R,G,X_3,Y_3,X^{[a+1]},Y^{[a+1]},X^{[a+2,3a]},Y^{[a+2,3a]},S$$ be  as defined in Algorithms~\ref{alg:2nmExtnew},~\ref{alg:2AdvCB},~\ref{alg:2FF}. We use the following sampling strategy:
\begin{enumerate}
    \item Sample  $g=(x_1,\bar{x}_2,y_1,\bar{y}_2) \leftarrow \tilde{G} =G=X_1 \circ \ecc(U_{n_2})_{\samp(\IP(X_1,Y_1))} \circ Y_1 \circ \ecc(U_{n_2})_{\samp(\IP(X_1,Y_1))}$,
    \item Sample $(x_3,x^{[a+1]},y_3,y^{[a+1]},s) \leftarrow \tilde{X}_3\tilde{X}^{[a+1]}\tilde{Y}_3\tilde{Y}^{[a+1]}\tilde{S}\vert (\tilde{G}=g)$ from Claim~\ref{claim:6final}, where $\tilde{S}\vert (\tilde{G} =g) =U_{n_x/4}.$
     \item Sample $(x^{[a+2,3a]},y^{[a+2,3a]}) \leftarrow \tilde{X}^{[a+2,3a]}\tilde{Y}^{[a+2,3a]}\vert (\tilde{G}=g,\tilde{X}_3\tilde{X}^{[a+1]}\tilde{Y}_3\tilde{Y}^{[a+1]}= x_3x^{[a+1]}y_3y^{[a+1]}) ={X}^{[a+2,3a]}{Y}^{[a+2,3a]}\vert ({G}=g,{X}_3{X}^{[a+1]}{Y}_3{Y}^{[a+1]}= x_3x^{[a+1]}y_3y^{[a+1]})$ from Claim~\ref{claim:new2nmextsampling1new}.
\end{enumerate}
Consider,\begin{align*}
        &\Vert XYS-\tilde{X}\tilde{Y}\tilde{S}\Vert_1\\ 
        & \leq \Vert GXYS-\tilde{G}\tilde{X}\tilde{Y}\tilde{S}\Vert_1& \mbox{(Fact~\ref{fact:data})}\\
        & \leq \Vert GX_1X_3X^{[a+1]}Y_1Y_3Y^{[a+1]}S -\tilde{G}\tilde{X}_1\tilde{X}_3\tilde{X}^{[a+1]}\tilde{Y}_1\tilde{Y}_3\tilde{Y}^{[a+1]}\tilde{S} \Vert_1& \mbox{(Fact~\ref{fact:data},~Step $(3)$~and~Claim~\ref{claim:new2nmextsampling1new})}\\
           & \leq \Vert GX_3X^{[a+1]}Y_3Y^{[a+1]}S -\tilde{G}\tilde{X}_3\tilde{X}^{[a+1]}\tilde{Y}_3\tilde{Y}^{[a+1]}\tilde{S} \Vert_1& \mbox{(Fact~\ref{fact:data})}\\
        & = \mathbb{E}_{g \leftarrow G}\Vert X_3X^{[a+1]}Y_3Y^{[a+1]}S \vert (G=g)-\tilde{X}_3\tilde{X}^{[a+1]}\tilde{Y}_3\tilde{Y}^{[a+1]}\tilde{S}\vert (\tilde{G}=g) \Vert_1& \mbox{(Fact~\ref{traceavg}~and~$\tilde{G}=G$)}\\
        & \leq \cO(\eps)& \mbox{(Claim~\ref{claim:6final})}.
        \end{align*}
\end{proof}

\begin{claim}\label{claim:new2nmextsampling1new}
  Let $2\nmext: \lbrace 0,1 \rbrace ^n\times \lbrace 0,1 \rbrace^n   \rightarrow \lbrace 0,1 \rbrace^{n_x/4}$ be the $\mathsf{new}$-$2\nmext$ from Algorithm~\ref{alg:2nmExtnew}. Let ${XY}=U_n \otimes U_n$, $S =2\nmext(X,Y)$ and intermediate random variables $$X_1,Y_1,G,X_3,Y_3,X^{[a]},Y^{[a]},X^{[a+1,3a]},Y^{[a+1,3a]}$$ be as defined in   Algorithms~\ref{alg:2nmExtnew},~\ref{alg:2AdvCB},~\ref{alg:2FF}. Then, we have 
  $${GX_3Y_3X^{[a+1]}Y^{[a+1]}}={G} \otimes {X_3Y_3X^{[a+1]}Y^{[a+1]}}.$$ 
  Furthermore, given $g \in \supp(G)$ and $x_3y_3x^{[a+1]}y^{[a+1]}$, we can efficiently (in time polynomial in $n$) sample from the distribution $X^{[a+2,3a]}Y^{[a+2,3a]} \vert (G=g,X_3Y_3X^{[a+1]}Y^{[a+1]}=x_3y_3y^{[a+1]}x^{[a+1]})$. 
\end{claim}
\begin{proof}
Note $X=X_1 \circ X_2$, $Y=Y_1 \circ Y_2$ and $G=X_1 \circ \ecc(X_2)_{\samp(\IP(X_1,Y_1))} \circ Y_1 \circ \ecc(Y_2)_{\samp(\IP(X_1,Y_1))}$ from the construction of $2\nmext$.
Note ${X_2Y_2}=U_{n_2} \otimes U_{n_2}$ and for fixed $G=g$, the distribution $X_2Y_2 \vert (G=g)=X_2 \vert (G=g) \otimes Y_2 \vert (G=g)$.

 Let $j \in [n_4]$ be such that $\pre(X_2,1,j \log (n+1))$ has the string $(X_3,X^{[a+1]})$ as prefix.  Thus, $(n_4-j) \log (n+1) \geq (2a-1)4n_x$. For our choice of parameters, $j < n_4 -n_5$. We now fix $G=g=x_1 \circ \bar{x_2} \circ y_1 \circ \bar{y_2}$ and let $\mathcal{T}_g =\samp(\IP(x_1,y_1)) =\{t_1,t_2, \ldots, t_{n_5}\}.$ 

Using Lemma~\ref{lem:reedsolomon}, with the following assignment of terms (terms on left are from Lemma~\ref{lem:reedsolomon} and terms on right are from here)
$$ (k,q,n,t,\mathcal{S}, \mathcal{Q},M) \leftarrow (n_4,n+1,n,n_5,\mathcal{T}_g,[j],X_2),$$ we get that $X_2 \vert (G=g)$ restricted to first $j \log (n+1)$ bits is uniform. Thus, $\pre(X_2\vert (G=g),1,j \log (n+1))$ is uniform. This further implies 
$(X_3,X^{[a+1]}) \vert (G=g)$ is uniform since  $(X_3,X^{[a+1]})$ is a prefix of $\pre(X_2\vert (G=g),1,j \log (n+1))$. Also, sampling from $X^{[a+2,3a]} \vert (G=g, X_3X^{[a+1]} =x_3x^{[a+1]})$ in time polynomial in $(n_4,n+1)$ follows from Lemma~\ref{lem:reedsolomon}. Using similar argument one can also note 
$(Y_3,Y^{[a+1]}) \vert (G=g)$ is uniform and sampling from $Y^{[a+2,3a]} \vert (G=g, Y_3Y^{[a+1]} =y_3y^{[a+1]})$ can be done in time polynomial in $(n_4,n+1)$. Thus, 
 ${GX_3Y_3X^{[a+1]}Y^{[a+1]}}={G} \otimes {X_3Y_3X^{[a+1]}Y^{[a+1]}}$ follows since  $X_3Y_3X^{[a+1]}Y^{[a+1]} \vert (G=g) =X_3Y_3X^{[a+1]}Y^{[a+1]} =U_{n_6} \otimes U_{n_6} \otimes U_{(a+1)4n_x} \otimes U_{(a+1)4n_x}$.

\end{proof}

\begin{claim}\label{claim:6final}
Let $2\nmext: \lbrace 0,1 \rbrace ^n\times \lbrace 0,1 \rbrace^n   \rightarrow \lbrace 0,1 \rbrace^{n_x/4}$ be the $\mathsf{new}$-$2\nmext$ from Algorithm~\ref{alg:2nmExtnew}. Let ${XY}=U_n \otimes U_n$, $S =2\nmext(X,Y)$ and
$$X_1,Y_1,R,G,X_3,Y_3,X^{[a]},Y^{[a]},Z^{[a]},A^{[a]},B^{[a]},C^{[a]},\bar{Z}^{[a]}, \bar{A}^{[a]}, \bar{B}^{[a]}, \bar{C}^{[a]},Z^{a+1},X^{[a+1,3a]},Y^{[a+1,3a]},F,S$$ be  as defined in Algorithms~\ref{alg:2nmExtnew},~\ref{alg:2AdvCB},~\ref{alg:2FF}. For any  $g \in \supp(G)$, we can sample efficiently from a distribution $(\tilde{X_3}\tilde{Y_3}\tilde{X}^{[a+1]}\tilde{Y}^{[a+1]}\tilde{S})\vert G=g$ such that 
\[\Vert (X_3Y_3X^{[a+1]}Y^{[a+1]}S) \vert (G=g)-(\tilde{X_3}\tilde{Y_3}\tilde{X}^{[a+1]}\tilde{Y}^{[a+1]}\tilde{S})\vert (G=g) \Vert_1 \leq {\cO(\eps)} \quad ;\quad \tilde{S}\vert (G=g) =U_{n_x/4}.\]

\end{claim}
\begin{proof}
Let $g=g_1g_2 \ldots g_a$. Note from Claim~\ref{claim:new2nmextsampling1new}, we have $GX_3Y_3X^{[a+1]}Y^{[a+1]}=G \otimes X_3Y_3X^{[a+1]}Y^{[a+1]}$. Thus, fixing $G=g$,  distribution of sources in the alternating extraction remains the same as before, i.e. exactly uniform. 

Suppose $g_a=1$, we show how to efficiently sample from distribution $(Z^aX^aY^aX^{a+1}Y^{a+1}S) \vert (G=g)$ approximately~\footnote{Similar argument holds when $g_a=0$.}. We remove conditioning of $G=g$ for the rest of the proof. From Algorithms~\ref{alg:2nmExtnew},~\ref{alg:2AdvCB}~and~\ref{alg:2FF}, we have
\[Y^a=Y^a_1 \circ Y^a_2\circ  Y^a_3\circ Y^a_4, \quad X^a=X^a_1 \circ X^a_2\circ X^a_3\circ X^a_4 \quad \textnormal{and} \quad \{ A^a,B^a,C^a,\bar{Z}^a,\bar{A}^a, Z^{a+1},F,S\}\]as intermediate random variables for $g_a=1$. From Claim~\ref{claim:newsampling}, we have $ \Vert \bar{Z}^a - U_{2h} \Vert_1 \leq \cO(a\eps')$. From construction of $2\nmext$, we have $ \bar{Z}^aY_3^{a}Y_4^{a}X_3^{a}X_4^{a}Y^{a+1}X^{a+1} =\bar{Z}^a \otimes Y_3^{a} \otimes Y_4^{a}\otimes X_3^{a}\otimes X_4^{a} \otimes Y^{a+1} \otimes X^{a+1} $. Thus, we have 
\begin{equation}\label{eq:finalnew1}\Vert\bar{Z}^aY_3^{a}Y_4^{a}X_3^{a}X_4^{a}Y^{a+1}X^{a+1} - U_{2h}\otimes U_{n_y} \otimes U_{n_y}\otimes U_{n_x} \otimes U_{n_x}\otimes U_{4n_y} \otimes U_{4n_x} \Vert_1 \leq \cO(a\eps').
\end{equation}Let us proceed flip-flop procedure (Algorithm~\ref{alg:2FF}) with $\hat{\bar{Z}}^a\hat{Y}_3^a\hat{Y}_4^a\hat{X}_3^a\hat{X}_4^a \hat{Y}^{a+1}\hat{X}^{a+1}=U_{2h}\otimes U_{n_y} \otimes U_{n_y}\otimes U_{n_x} \otimes U_{n_x}\otimes U_{4n_y} \otimes U_{4n_x}$ and denote intermediate random variables as hat variables. Using Fact~\ref{fact:data} and Eq.~\eqref{eq:finalnew1}, we have 
\begin{equation}\label{eq:finalnew2}\Vert \bar{Z}^aY_3^{a}Y_4^{a}X_3^{a}X_4^{a}Y^{a+1}X^{a+1}\bar{A}^aZ^{a+1}FS - \hat{\bar{Z}}^a\hat{Y}_3^a\hat{Y}_4^a\hat{X}_3^a\hat{X}_4^a\hat{Y}^{a+1}\hat{X}^{a+1}\hat{\bar{A}}^a\hat{Z}^{a+1}\hat{F}\hat{S} \Vert_1 \leq \cO(a\eps').
\end{equation}
From Algorithms~\ref{alg:2nmExtnew},~\ref{alg:2AdvCB},~\ref{alg:2FF}, we have 
\begin{equation}\label{eq:claim8num1}
     (\hat{\bar{Z}}_s^a \otimes \hat{Y}_3^a) \leftrightarrow (\hat{\bar{A}}^a \otimes \hat{X}_4^a) \leftrightarrow (\hat{Z}^{a+1} \otimes \hat{Y}^{a+1}) \leftrightarrow (\hat{F} \otimes \hat{X}^{a+1}) \leftrightarrow \hat{S}.
\end{equation}
Using Corollary~\ref{corr:claim4} and noting $\hat{\bar{Z}}_s^a\hat{Y}_3^a=U_{s}\otimes U_{n_y}$, we have, $\Vert \hat{\bar{A}}^a -U_b \Vert_1 \leq \eps'.$ Noting $\hat{\bar{A}}^a \hat{X}_4^{a}= \hat{\bar{A}}^a \otimes  U_{n_x}$, we further have 
\[\Vert \hat{\bar{A}}^a \hat{X}_4^{a} -U_b \otimes U_{n_x}\Vert_1 \leq \eps'.\]
Using similar arguments, we also get 
\[\Vert \hat{Z}^{a+1}\hat{Y}^{a+1} -U_{2h} \otimes U_{4n_y}\Vert_1 \leq 2\eps' \quad ;\quad \Vert  \hat{F} \hat{X}^{a+1} -U_{n_y/8} \otimes U_{4n_x}\Vert_1 \leq 2\eps'+\eps^2.\]

We use the following sampling strategy to sample from a distribution $\tilde{\bar{Z}}^a\tilde{Y}_3^a\tilde{Y}_4^a \tilde{\bar{A}}^a\tilde{X}_3^{a}\tilde{X}_4^{a}\tilde{Z}^{a+1} \tilde{Y}^{a+1}\tilde{F}\tilde{X}^{a+1}\tilde{S}$:
\begin{enumerate}
    \item Sample  $(\tilde{\bar{z}}^a,\tilde{\bar{a}}^a,\tilde{z}^{a+1},\tilde{f},\tilde{s}) \leftarrow \tilde{\bar{Z}}^a\tilde{\bar{A}}^a\tilde{Z}^{a+1}\tilde{F}\tilde{S} = U_{2h} \otimes U_b \otimes  U_{2h} \otimes U_{n_y/8} \otimes U_{n_x/4}$,
    \item Sample $\tilde{x}^{a+1} \leftarrow  \tilde{X}^{a+1}\vert (\tilde{F}\tilde{S}=\tilde{f}\tilde{s})$ which is same as $ \tilde{x}^{a+1} \leftarrow \{ {x}^{a+1}  : \tilde{s}= \Ext_6({x}^{a+1},\tilde{f})\}$,
     \item Sample $\tilde{y}^{a+1} \leftarrow  \tilde{Y}^{a+1}\vert (\tilde{Z}^{a+1}\tilde{F}=\tilde{z}^{a+1}\tilde{f})$ which is same as $\tilde{y}^{a+1} \leftarrow \{ {y}^{a+1}  : \tilde{f}= \Ext_4({y}^{a+1},\tilde{z}^{a+1})\}$,
    \item Sample $\tilde{x}_4^{a} \leftarrow \tilde{X}_4^{a}\vert (\tilde{\bar{A}}^a\tilde{Z}^{a+1}=\tilde{\bar{a}}^a\tilde{z}^{a+1})$ which is same as $\tilde{x}_4^{a} \leftarrow \{ {x}_4^{a}  : \tilde{z}^{a+1}= \Ext_3({x}_4^{a},\tilde{\bar{a}}^a)\}$,
     \item Sample $\tilde{y}_3^{a} \leftarrow \tilde{Y}_3^{a}\vert (\tilde{\bar{Z}}_s^a\tilde{\bar{A}}^a=\tilde{\bar{z}}_s^a\tilde{\bar{a}}^a)$ which is same as $\tilde{y}_3^{a} \leftarrow \{ {y}_3^{a}  : \tilde{\bar{a}}^a= \Ext_1({y}_3^{a},\tilde{\bar{z}}_s^a)\}$,
     \item Sample $\tilde{y}_4^{a} \leftarrow \tilde{Y}_4^{a}= U_{n_y}$ and independently of $\tilde{\bar{Z}}^a\tilde{Y}_3^a\ \tilde{\bar{A}}^a\tilde{X}_4^{a}\tilde{Z}^{a+1} \tilde{Y}^{a+1}\tilde{F}\tilde{X}^{a+1}\tilde{S}$,
     \item Sample $\tilde{x}_3^{a} \leftarrow \tilde{X}_3^{a}= U_{n_x}$ and independently of $\tilde{\bar{Z}}^a\tilde{Y}_3^a\tilde{Y}_4^a\ \tilde{\bar{A}}^a\tilde{X}_4^{a}\tilde{Z}^{a+1} \tilde{Y}^{a+1}\tilde{F}\tilde{X}^{a+1}\tilde{S}$.
\end{enumerate}Noting,
\[\Vert \hat{F} \hat{X}^{a+1} -U_{n_y/8} \otimes U_{4n_x}\Vert_1 \leq 2\eps'+\eps^2 \quad ;\quad \hat{S}=\Ext_6(\hat{X}^{a+1},\hat{F})\] and using Claim~\ref{claim:extractorsampling}, 
 with the below assignment of terms (terms on the left are from Claim~\ref{claim:extractorsampling} and on the right are from here), $$(X,H,O,\tilde{X},\tilde{H},\tilde{O},\eps'') \leftarrow (\hat{X}^{a+1},\hat{F},\hat{S},\tilde{X}^{a+1},\tilde{F},\tilde{S},0),$$we get that the sampled distribution $\tilde{F}\tilde{X}^{a+1}\tilde{S}$ satisfies
\begin{equation}\label{eq:claimfinal1}
    \Vert \hat{F}\hat{X}^{a+1}\hat{S} -\tilde{F}\tilde{X}^{a+1}\tilde{S}\Vert_1 \leq 2\eps'+2\eps^2.
\end{equation}Noting,
\[\Vert \hat{Z}^{a+1}\hat{Y}^{a+1} -U_{2h} \otimes U_{4n_y}\Vert_1 \leq 2\eps' \quad ;\quad \hat{F}=\Ext_4(\hat{Y}^{a+1},\hat{Z}^{a+1})\] and using Claim~\ref{claim:extractorsampling} again, 
 with the below assignment of terms (terms on the left are from Claim~\ref{claim:extractorsampling} and on the right are from here), $$(X,H,O,\tilde{X},\tilde{H},\tilde{O},\eps'') \leftarrow (\hat{Y}^{a+1},\hat{Z}^{a+1},\hat{F},\tilde{Y}^{a+1},\tilde{Z}^{a+1},\tilde{F},0),$$we get that the sampled distribution  $\tilde{Z}^{a+1}\tilde{Y}^{a+1}\tilde{F}$ satisfies
\begin{equation}\label{eq:claimfinal2}
   \Vert \hat{Z}^{a+1}\hat{Y}^{a+1}\hat{F} -\tilde{Z}^{a+1}\tilde{Y}^{a+1}\tilde{F}\Vert_1 \leq 2\eps'+\eps^2.
\end{equation}Note $\hat{Z}^{a+1}\hat{Y}^{a+1} \leftrightarrow \hat{F} \leftrightarrow \hat{X}^{a+1} \hat{S}$ (from Eq.~\eqref{eq:claim8num1}), $ \tilde{Z}^{a+1}\tilde{Y}^{a+1} \leftrightarrow \tilde{F} \leftrightarrow \tilde{X}^{a+1} \tilde{S}$ (from the sampling procedure). Thus, from Eqs.~\eqref{eq:claimfinal1},~\eqref{eq:claimfinal2} and using Claim~\ref{fact:markovapprox} with the below assignment of terms (terms on the left are from Claim~\ref{fact:markovapprox} and on the right are from here), $$(A,B,C,\tilde{A},\tilde{B},\tilde{C}) \leftarrow (\hat{Z}^{a+1}\hat{Y}^{a+1}, \hat{F},\hat{X}^{a+1} \hat{S},\tilde{Z}^{a+1}\tilde{Y}^{a+1}, \tilde{F},\tilde{X}^{a+1} \tilde{S}),$$we get that the sampled distribution $\tilde{Z}^{a+1}\tilde{Y}^{a+1}\tilde{F}\tilde{X}^{a+1}\tilde{S}$ satisfies
\begin{equation}\label{eq:claimfinal3}
  \Vert \hat{Z}^{a+1}\hat{Y}^{a+1}\hat{F}\hat{X}^{a+1}\hat{S} -\tilde{Z}^{a+1}\tilde{Y}^{a+1}\tilde{F}\tilde{X}^{a+1}\tilde{S}\Vert_1 \leq 6\eps'+4\eps^2\leq 6\eps'+4\eps .
\end{equation}
Noting,
\[\Vert \hat{\bar{A}}^a \hat{X}_4^{a} -U_{b} \otimes U_{n_x}\Vert_1 \leq \eps' \quad ;\quad \hat{Z}^{a+1}=\Ext_3(\hat{X}_4^a,\hat{\bar{A}}^a)\]  and using Claim~\ref{claim:extractorsampling} again, 
 with the below assignment of terms $$(X,H,O,\tilde{X},\tilde{H},\tilde{O},\eps'') \leftarrow (\hat{X}_4^{a},\hat{\bar{A}}^{a},\hat{Z}^{a+1},\tilde{X}_4^{a},\tilde{\bar{A}}^{a},\tilde{Z}^{a+1},0),$$we get that the sampled distribution $\tilde{\bar{A}}^{a}\tilde{X_4}^{a}\tilde{Z}^{a+1}$ satisfies \begin{equation}\label{eq:claimfinal4}
   \Vert \hat{\bar{A}}^a \hat{X}_4^{a}\hat{Z}^{a+1} -\tilde{\bar{A}}^{a}\tilde{X_4}^{a}\tilde{Z}^{a+1}\Vert_1 \leq 2\eps'.
\end{equation}
Note $ \hat{\bar{A}}^a \hat{X}_4^{a} \leftrightarrow \hat{Z}^{a+1} \leftrightarrow \hat{Y}^{a+1}\hat{F}\hat{X}^{a+1}\hat{S}$ (from Eq.~\eqref{eq:claim8num1}) and $ \tilde{\bar{A}}^a \tilde{X}_4^{a} \leftrightarrow \tilde{Z}^{a+1} \leftrightarrow \tilde{Y}^{a+1}\tilde{F}\tilde{X}^{a+1}\tilde{S}$ (from sampling procedure). Thus, from Eqs.~\eqref{eq:claimfinal3},~\eqref{eq:claimfinal4} and using Claim~\ref{fact:markovapprox} with the below assignment of registers $$(A,B,C,\tilde{A},\tilde{B},\tilde{C}) \leftarrow ( \hat{\bar{A}}^a \hat{X}_4^{a},\hat{Z}^{a+1}, \hat{Y}^{a+1}\hat{F}\hat{X}^{a+1}\hat{S},\tilde{\bar{A}}^a \tilde{X}_4^{a},\tilde{Z}^{a+1}, \tilde{Y}^{a+1}\tilde{F}\tilde{X}^{a+1}\tilde{S}),$$ we get that the sampled distribution $\tilde{\bar{A}}^a\tilde{X}_4^{a}\tilde{Z}^{a+1}\tilde{Y}^{a+1}\tilde{F}\tilde{X}^{a+1}\tilde{S}$ satisfies
\begin{equation}\label{eq:claimfinal5}
  \Vert \hat{\bar{A}}^a \hat{X}_4^{a}\hat{Z}^{a+1} \hat{Y}^{a+1}\hat{F}\hat{X}^{a+1}\hat{S} - \tilde{\bar{A}}^a\tilde{X}_4^{a}\tilde{Z}^{a+1}\tilde{Y}^{a+1}\tilde{F}\tilde{X}^{a+1}\tilde{S}\Vert_1 \leq 10\eps'+4\eps.
\end{equation}
Proceeding similarly using arguments involving Claim~\ref{claim:extractorsampling}, Claim~\ref{fact:markovapprox} appropriately, the sampled distribution $\tilde{\bar{Z}}_s^a\tilde{Y}_3^a \tilde{\bar{A}}^a \tilde{X}_4^{a}\tilde{Z}^{a+1} \tilde{Y}^{a+1}\tilde{F}\tilde{X}^{a+1}\tilde{S}$ satisfies,
\begin{equation}\label{eq:finalnew6}\Vert \hat{\bar{Z}}_s^a\hat{Y}_3^a \hat{\bar{A}}^a \hat{X}_4^{a}\hat{Z}^{a+1} \hat{Y}^{a+1}\hat{F}\hat{X}^{a+1}\hat{S}  - \tilde{\bar{Z}}_s^a\tilde{Y}_3^a \tilde{\bar{A}}^a \tilde{X}_4^{a}\tilde{Z}^{a+1} \tilde{Y}^{a+1}\tilde{F}\tilde{X}^{a+1}\tilde{S}  \Vert_1 \leq 12\eps'+4\eps.
\end{equation}
Also note registers $ \pre( \hat{\bar{Z}}^a,s+1,2h )\hat{X}_3^a\hat{Y}_4^a = U_{2h-s}\otimes U_{n_x} \otimes U_{n_y}$ and independent of registers $\hat{\bar{Z}}_s^a\hat{Y}_3^a \hat{\bar{A}}^a \hat{X}_4^{a}\hat{Z}^{a+1} \hat{Y}^{a+1}\hat{F}\hat{X}^{a+1}\hat{S}$ (similarly for tilde variables). Thus,
\begin{equation}\label{eq:finalnew16}\Vert \hat{\bar{Z}}^a\hat{Y}_3^a\hat{Y}_4^a \hat{\bar{A}}^a\hat{X}_3^{a} \hat{X}_4^{a}\hat{Z}^{a+1} \hat{Y}^{a+1}\hat{F}\hat{X}^{a+1}\hat{S}  - \tilde{\bar{Z}}^a\tilde{Y}_3^a\tilde{Y}_4^a \tilde{\bar{A}}^a\tilde{X}_3^{a} \tilde{X}_4^{a}\tilde{Z}^{a+1} \tilde{Y}^{a+1}\tilde{F}\tilde{X}^{a+1}\tilde{S}  \Vert_1 \leq 12\eps'+4\eps.
\end{equation}
Using Eqs.~\eqref{eq:finalnew2},~\eqref{eq:finalnew16} and triangle inequality, we get\begin{equation*}\Vert  \bar{Z}^aY_3^{a}Y_4^{a}\bar{A}^aX_3^{a}X_4^{a}Z^{a+1}Y^{a+1}FX^{a+1}S -\tilde{\bar{Z}}^a\tilde{Y}_3^a\tilde{Y}_4^a \tilde{\bar{A}}^a\tilde{X}_3^{a} \tilde{X}_4^{a}\tilde{Z}^{a+1} \tilde{Y}^{a+1}\tilde{F}\tilde{X}^{a+1}\tilde{S}  \Vert_1 \leq \cO((a+1)\eps') +4\eps.
\end{equation*}Using Fact~\ref{fact:data}, we further have that the sampled distribution $\tilde{\bar{Z}}^a\tilde{X}_3^a \tilde{X}_4^{a}\tilde{Y}_3^a \tilde{Y}_4^{a}\tilde{X}^{a+1} \tilde{Y}^{a+1}\tilde{S}$ satisfies
\begin{equation}\label{eq:finalnew7}\Vert  \bar{Z}^aX_3^{a}X_4^{a}Y_3^{a}Y_4^{a}X^{a+1}Y^{a+1}S -\tilde{\bar{Z}}^a\tilde{X}_3^a\tilde{X}_4^a\tilde{Y}_3^{a} \tilde{Y}_4^{a} \tilde{X}^{a+1}\tilde{Y}^{a+1}\tilde{S}  \Vert_1 \leq \cO((a+1)\eps') +4\eps.
\end{equation}

From Claim~\ref{claim:newsampling}, we have $ \Vert Z^a - U_{2h} \Vert_1 \leq \cO(a\eps')$. From construction of $2 \nmext$, we have $ {Z}^aY_1^{a}Y_2^{a}X_1^{a}X_2^{a}={Z}^a \otimes Y_1^{a} \otimes Y_2^{a}\otimes X_1^{a}\otimes X_2^{a}$. Thus, we have 
\begin{equation}\label{eq:finalnew11}\Vert{Z}^aY_1^{a}Y_2^{a}X_1^{a}X_2^{a}- U_{2h}\otimes U_{n_y} \otimes U_{n_y}\otimes U_{n_x} \otimes U_{n_x}  \Vert_1 \leq \cO(a\eps').
\end{equation}Proceed flip-flop procedure with $\hat{Z}^a\hat{X}_1^a\hat{X}_2^a\hat{Y}_1^a\hat{Y}_2^a = U_{2h}\otimes U_{n_x} \otimes U_{n_x}\otimes U_{n_y} \otimes U_{n_y}$ and denote intermediate random variables as hat variables. Using Fact~\ref{fact:data} and Eq.~\eqref{eq:finalnew11}, we have 
\begin{equation}\label{eq:finalnew22}\Vert {Z}^aY_1^{a}Y_2^{a}X_1^{a}X_2^{a}{A}^aC^aB^a\bar{Z}^a -\hat{Z}^a\hat{Y}_1^{a}\hat{Y}_2^{a}\hat{X}_1^{a}\hat{X}_2^{a}\hat{A}^a\hat{C}^a\hat{B}^a\hat{\bar{Z}}^a \Vert_1 \leq \cO(a\eps').
\end{equation}From Algorithms~\ref{alg:2nmExtnew},~\ref{alg:2AdvCB},~\ref{alg:2FF}, we have 
\begin{equation}\label{eq:claim8numdfg}
     (\hat{{Z}}_s^a \otimes \hat{Y}_1^a) \leftrightarrow (\hat{{A}}^a \otimes \hat{{Z}}_2^a) \leftrightarrow (\hat{C}^{a} \otimes \hat{Y}_2^{a}) \leftrightarrow (\hat{B}^a \otimes \hat{X}_2^{a}) \leftrightarrow \hat{\bar{Z}}^a.
\end{equation}
We extend the sampling procedure in the following way: consider sampled $\tilde{\bar{z}}^a \leftarrow \tilde{\bar{Z}}^a=U_{2h}$ (from before), next sample $\tilde{b}^a \leftarrow \tilde{B}^a=U_b$ and independently and proceed to sample from $\tilde{X}_2^a \vert ( \tilde{\bar{Z}}^a\tilde{B}^a= \tilde{\bar{z}}^a\tilde{b}^a)$.
Using similar arguments involving Corollary~\ref{corr:claim4}, Claim~\ref{claim:extractorsampling}, Claim~\ref{fact:markovapprox} like before (in this proof), we can sample from a  distribution  $\tilde{Z}^a\tilde{Y}_1^{a}\tilde{Y}_2^{a}\tilde{X}_1^{a}\tilde{X}_2^{a}\tilde{A}^a\tilde{C}^a\tilde{B}^a\tilde{\bar{Z}}^a$ such that 
\begin{equation}\label{eq:finalnew222}\Vert \hat{Z}^a\hat{Y}_1^{a}\hat{Y}_2^{a}\hat{X}_1^{a}\hat{X}_2^{a}\hat{A}^a\hat{C}^a\hat{B}^a\hat{\bar{Z}}^a -\tilde{Z}^a\tilde{Y}_1^{a}\tilde{Y}_2^{a}\tilde{X}_1^{a}\tilde{X}_2^{a}\tilde{A}^a\tilde{C}^a\tilde{B}^a\tilde{\bar{Z}}^a \Vert_1 \leq \cO(\eps').
\end{equation}
Using Eq.~\eqref{eq:finalnew222} along with Fact~\ref{fact:data} we have, 
\begin{equation}\label{eq:finalnew189}\Vert   \tilde{{Z}}^a -U_{2h} \Vert_1 \leq \cO(\eps').
\end{equation}From Eqs.~\eqref{eq:finalnew22},~\eqref{eq:finalnew222} and triangle inequality,  we have
\begin{equation}\label{eq:finalnew23}\Vert {Z}^aY_1^{a}Y_2^{a}X_1^{a}X_2^{a}{A}^aC^aB^a\bar{Z}^a -\tilde{Z}^a\tilde{Y}_1^{a}\tilde{Y}_2^{a}\tilde{X}_1^{a}\tilde{X}_2^{a}\tilde{A}^a\tilde{C}^a\tilde{B}^a\tilde{\bar{Z}}^a \Vert_1 \leq \cO((a+1)\eps').
\end{equation}Using Fact~\ref{fact:data}, we further can sample from a distribution $\tilde{Z}^a\tilde{Y}_1^{a}\tilde{Y}_2^{a}\tilde{X}_1^{a}\tilde{X}_2^{a}\tilde{\bar{Z}}^a$ such that
\begin{equation}\label{eq:finalnew77}\Vert {Z}^aY_1^{a}Y_2^{a}X_1^{a}X_2^{a}\bar{Z}^a -\tilde{Z}^a\tilde{Y}_1^{a}\tilde{Y}_2^{a}\tilde{X}_1^{a}\tilde{X}_2^{a}\tilde{\bar{Z}}^a \Vert_1 \leq \cO((a+1)\eps').
\end{equation}
Note that ${Z}^aY_1^{a}Y_2^{a}X_1^{a}X_2^{a} \leftrightarrow \bar{Z}^a \leftrightarrow X_3^{a}X_4^{a}Y_3^{a}Y_4^{a}X^{a+1}Y^{a+1}S$ (from construction of $2\nmext$)
 and $\tilde{Z}^a\tilde{Y}_1^{a}\tilde{Y}_2^{a}\tilde{X}_1^{a}\tilde{X}_2^{a} \leftrightarrow \tilde{\bar{Z}}^a \leftrightarrow \tilde{X}_3^{a}\tilde{X}_4^{a}\tilde{Y}_3^{a}\tilde{Y}_4^{a}\tilde{X}^{a+1}\tilde{Y}^{a+1}\tilde{S}$ (from sampling procedure). From  Eqs.~\eqref{eq:finalnew7},~\eqref{eq:finalnew77} and using  Claim~\ref{fact:markovapprox} along with Fact~\ref{fact:data}, we have 
\[ \Vert {Z}^aY^{[a,a+1]}X^{[a,a+1]}S -\tilde{Z}^a\tilde{Y}^{[a,a+1]}\tilde{X}^{[a,a+1]}\tilde{S} \Vert_1 \leq \cO((a+1)\eps')+ 4\eps.\]

We extend the sampling procedure now considering sampled $\tilde{z}^a \leftarrow \tilde{{Z}}^a$ (from before) such that $\Vert \tilde{{Z}}^a -U_{2h} \Vert_1 \leq \cO(\eps')$ (from Eq.~\eqref{eq:finalnew189}) and proceed depending on bit $g_{a-1}$. Proceeding likewise depending on $g=g_1g_2\ldots g_a$ and  inductively using similar arguments involving Corollary~\ref{corr:claim4}, Claim~\ref{claim:extractorsampling}, Claim~\ref{fact:markovapprox} along with Fact~\ref{fact:data}, we can sample from $\tilde{X}_3\tilde{Y}_3\tilde{X}^{[a+1]}\tilde{Y}^{[a+1]}\tilde{S}$ such that 
\begin{equation*}\label{eq:finalnew2222}\Vert X_3Y_3{X}^{[a+1]}{Y}^{[a+1]}{S} -\tilde{X}_3\tilde{Y}_3\tilde{X}^{[a+1]}\tilde{Y}^{[a+1]}\tilde{S} \Vert_1 \leq \cO(\eps') + \left(\sum_{i=1}^{a} \cO((a+1)\eps')\right)+ 4\eps = \cO( \eps).
\end{equation*}
\end{proof}

We show a claim that states that intermediate random variables in the alternating extraction are approximately uniform even conditioned on every $G=g$.
\begin{claim}\label{claim:newsampling}
 Let $2\nmext: \lbrace 0,1 \rbrace ^n\times \lbrace 0,1 \rbrace^n   \rightarrow \lbrace 0,1 \rbrace^{n_x/4}$ be the $\mathsf{new}$-$2\nmext$ from Algorithm~\ref{alg:2nmExtnew}. Let ${XY}=U_n \otimes U_n$, $S =2\nmext(X,Y)$ and set $$ \mathsf{P} \defeq\{G, Z^{[a]},A^{[a]},B^{[a]},C^{[a]},\bar{Z}^{[a]}, \bar{A}^{[a]}, \bar{B}^{[a]}, \bar{C}^{[a]},Z^{a+1}\}$$ be the  intermediate random variables as defined in   Algorithms~\ref{alg:2nmExtnew},~\ref{alg:2AdvCB},~\ref{alg:2FF}. Then, we have for any random variable $Q \in \mathsf{P}\setminus \{G \}$ and any fixing $G=g$, 
 $$ \Vert Q \vert (G=g)-U_{\vert Q \vert} \Vert_1 \leq  \cO(a\eps').$$
\end{claim}
\begin{proof}
From Claim~\ref{claim:new2nmextsampling1new}, we have  $${GX_3Y_3X^{[a+1]}Y^{[a+1]}}={G} \otimes {X_3Y_3X^{[a+1]}Y^{[a+1]}},$$
where $X=X_1 \circ X_3 \circ X^{[a+1]} \circ X^{[a+2,3a]}$ and $Y=Y_1 \circ Y_3 \circ Y^{[a+1]} \circ Y^{[a+2,3a]}$. Also, note from Algorithms~\ref{alg:2nmExtnew},~\ref{alg:2AdvCB},~\ref{alg:2FF}, any random variable $Q \in \mathsf{P}\setminus \{G \}$ is extracted from sources $X_3Y_3X^{[a+1]}Y^{[a+1]}$. Note for any $i \in [a]$ and $i\mhyphen$th flip-flop procedure (Algorithm~\ref{alg:2FF}), intermediate random variables $ \{ A^{i},B^{i},C^{i},\bar{Z}^{i}, \bar{A}^{i}, \bar{B}^{i}, \bar{C}^{i},Z^{i+1} \}$ are extracted from $X^iY^iZ^i =X^i \otimes Y^i \otimes Z^i$. We remove conditioning on $G=g$ for the random variables for the rest of the proof. 

Since $Z^1 =\IP(X_3,Y_3)$, using Fact~\ref{classicalip} we have $\Vert Z^1 - U_{2h} \Vert_1 \leq 2^{-\Omega(n)} \leq \eps'$. We show if for any round $i$, $\Vert Z^i - U_{2h} \Vert_1 \leq \cO(i\eps')$, then for any $Q \in  \{ A^{i},B^{i},C^{i},\bar{Z}^{i}, \bar{A}^{i}, \bar{B}^{i}, \bar{C}^{i},Z^{i+1} \}$ we have 
$\Vert Q - U_{\vert Q \vert} \Vert_1 \leq \cO((i+2)\eps')$.

We denote  generated intermediate random variables $A^i,B^i,C^i,\bar{Z}^{i}$ from $\hat{Z}^i=\hat{Z}^i_1\hat{Z}^i_2=U_{2h}$ instead of $Z^i$ as hat random variables. Using Fact~\ref{fact:data} repeatedly, we have 
\begin{equation}\label{eq:claimseedepsilon}\Vert X_1^iX_2^iY_1^iY_2^iZ^i A^iB^iC^i\bar{Z}^{i} -X_1^iX_2^iY_1^iY_2^i \hat{Z}^i \hat{A}^i\hat{B}^i\hat{C}^i\hat{\bar{Z}}^{i}  \Vert_1 \leq  \cO(i\eps').
\end{equation}
Since $\hat{A}^i = \Ext_1({Y}^i_1,\hat{Z}^i_s)$, ${Y}^i_1\hat{Z}^i_s = U_{n_y} \otimes U_s$,  using Corollary~\ref{corr:claim4} we have 
\begin{equation}\label{eq:claimseedepsilon1}\Vert \hat{A}^i -U_b\Vert_1\leq \eps'.
\end{equation}
Since $\hat{C}^i = \Ext_2(\hat{Z}^i_2,\hat{A}^i)$, $\hat{Z}^i_2\hat{A}^i=U_h \otimes \hat{A}^i$ and $\Vert \hat{A}^i -U_b\Vert_1\leq \eps'$, using Corollary~\ref{corr:claim4} we have 
\begin{equation}\label{eq:claimseedepsilon2}\Vert \hat{C}^i -U_s\Vert_1\leq 2\eps'.
\end{equation}
Since $\hat{B}^i = \Ext_1({Y}^i_2,\hat{C}^i)$, ${Y}^i_2\hat{C}^i=U_{n_y} \otimes \hat{C}^i$ and $\Vert \hat{C}^i -U_s\Vert_1\leq 2\eps'$, using Corollary~\ref{corr:claim4} we have 
\begin{equation}\label{eq:claimseedepsilon3}\Vert \hat{B}^i -U_b\Vert_1\leq 3\eps'.
\end{equation}
If $g_i =1$, we have $\hat{\bar{Z}}^i = \Ext_1({X}^i_2,\hat{B}^i)$, ${X}^i_2\hat{B}^i=U_{n_x} \otimes \hat{B}^i$ and $\Vert \hat{B}^i -U_b\Vert_1\leq 3\eps'$. Using Corollary~\ref{corr:claim4}, we further have\begin{equation}\label{eq:claimseedepsilon4}\Vert \hat{\bar{Z}}^i -U_{2h}\Vert_1\leq 4\eps'.
\end{equation}Similarly, if $g_i =0$, we will have
\begin{equation}\label{eq:claimseedepsilon5}\Vert \hat{\bar{Z}}^i -U_{2h}\Vert_1\leq 2\eps'.
\end{equation}
Using Fact~\ref{fact:data}, Eqs.~\eqref{eq:claimseedepsilon},~\eqref{eq:claimseedepsilon1},~\eqref{eq:claimseedepsilon2},~\eqref{eq:claimseedepsilon3},~\eqref{eq:claimseedepsilon4}~and~\eqref{eq:claimseedepsilon5} along with appropriate triangle inequalities, for any $Q \in  \{ A^{i},B^{i},C^{i},\bar{Z}^{i}\}$ we have 
$\Vert Q - U_{\vert Q \vert} \Vert_1 \leq \cO((i+1)\eps')$. Using similar arguments and noting 
$X_3^iX_4^iY_3^iY_4^i\bar{Z}^i =X_3^i\otimes X_4^i\otimes Y_3^i\otimes Y_4^i\otimes \bar{Z}^i,$ one can note for any $Q \in  \{  \bar{A}^{i}, \bar{B}^{i}, \bar{C}^{i},Z^{i+1} \}$ we have 
$\Vert Q - U_{\vert Q \vert} \Vert_1 \leq \cO((i+2)\eps')$. Thus, since we have $a$ rounds of flip-flop procedure, the desired follows from induction argument.

\end{proof}


\subsection*{Acknowledgment}
We thank Divesh Aggarwal and Maciej Obremski for introducing us to the problem, sharing their insights on the classical constructions and several other helpful discussions. We also thank Upendra Kapshikar for the helpful discussions. 

This work is supported by the National Research Foundation, including under NRF RF Award
No. NRF-NRFF2013-13, the Prime Minister’s Office, Singapore and the Ministry of Education,
Singapore, under the Research Centres of Excellence program and by Grant No. MOE2012-T3-1-
009 and in part by the NRF2017-NRF-ANR004 VanQuTe Grant. This work is also partially supported by the Singapore Ministry of Education under grant MOE2019-T2-1-145.

\bibliography{References}
\bibliographystyle{alpha}
\appendix

\section{A quantum secure one-many non-malleable code in the split-state model \label{Appendix1}}

\subsection{Preliminaries for quantum secure one-many non-malleable code in the split-state model}
 Let $m,n,k,t$ be positive integers and $\eps,\eps'>0$. Let $\mathcal{F}$ denote the set of all the functions from $\{0,1 \}^n \longrightarrow \{0,1 \}^n$. We next define $\cp^{(t)} : (\{ 0,1\}^* \cup \{ \mathsf{same}\})^t \times \{ 0,1\}^* \to \{ 0,1\}^*$ when the first input is a $t$-tuple as follows:
$\cp^{(t)}((x_1, x_2, \ldots, x_t),s) = (\cp(x_1,s), \cp(x_2,s), \ldots, \cp(x_t,s))$. Recall the function $\cp : \{ 0,1\}^* \cup \{ \mathsf{same}\} \times \{ 0,1\}^* \to \{ 0,1\}^*$ is such that $\cp(x,s) =s$ if $x=\mathsf{same}$, otherwise $\cp(x,s) =x$.

\begin{definition}[One-many non-malleable codes in the split-state model\label{def:nmcodes_classicalt}~\cite{CGL15}]An encoding and decoding scheme $(\enc, \dec)$ is a $(t;m,n,\eps)$-non-malleable code with respect to a family of tampering functions ${(\mathcal{F}\times \mathcal{F})}^{t} $, if for every $$f=((g_1,h_1),(g_2,h_2), \ldots, (g_t,h_t)) \in {(\mathcal{F}\times \mathcal{F})}^{ t} ,$$ there exists a random variable $\mathcal{D}_{f}=\mathcal{D}_{((g_1,h_1),(g_2,h_2), \ldots, (g_t,h_t))}$ on $ (\{0,1 \}^{m}  \cup \lbrace \sm \rbrace)^{ t}$ which is independent of the randomness in $\enc$ such that for all messages $s \in \{0,1 \}^m$, it holds that
$$  \Vert  \dec((g_1,h_1) (\enc(s))) \ldots  \dec((g_t,h_t)  (\enc(s)))  - \cp^{(t)} (\mathcal{D}_{f} ,s)\Vert_1 \leq \eps.$$
\end{definition}

\begin{definition}[$(t;k_1,k_2)\mhyphen\nmas$~\cite{BJK21}]\label{def:2tsource-qnmadversarydef}
     Let $\sigma_{{X}\hat{X}NMY\hat{Y}}$ be a $(k_1,k_2)\mhyphen\qmas$. Let $U: \cH_X \otimes \cH_N \rightarrow \cH_X \otimes \cH_{X^{[t]}}  \otimes  \cH_{\hat{X}^{[t]}}  \otimes \cH_{N'}$ and $V: \cH_Y \otimes \cH_M \rightarrow \cH_Y \otimes \cH_{Y^{[t]}}  \otimes  \cH_{\hat{Y}^{[t]}}  \otimes \cH_{M'}$ be isometries such that for
     $\rho = (U \otimes V)\sigma (U \otimes V)^\dagger,$ we have $X^{[t]}Y^{[t]}$ classical (with copy $\hat{X}^{[t]}\hat{Y}^{[t]}$) and
      \begin{gather*}
           \forall i\in [t],\ \Pr(Y \ne Y^i)_\rho =1 \quad \text{or} \quad \Pr(X \ne X^i)_\rho =1. 
       \end{gather*}We call $\rho$ a $(t;k_1,k_2)\mhyphen\nmas$.
\end{definition}
\begin{definition}[Quantum secure $2$-source  one-many non-malleable extractor~\cite{BJK21}]\label{def:2tsourcenme}
		An $(n,n,m)$-non-malleable extractor $t\mhyphen2\nmext : \{0,1\}^{n} \times \{0,1\}^{n} \to \{0,1\}^m$ is $(t;k_1,k_2,\eps)$-secure against $\nma$ if for every $(t;k_1,k_2)\mhyphen\nmas$ $\rho$ (chosen by the adversary $\nma$),
	$$  \| \rho_{ t\mhyphen2\nmext(X,Y)t\mhyphen2\nmext(X^1,Y^1) \ldots t\mhyphen2\nmext(X^t,Y^t) XX^{[t]}N'} - U_m \otimes \rho_{t\mhyphen2\nmext(X^1,Y^1) \ldots t\mhyphen2\nmext(X^t,Y^t) XX^{[t]}N'} \|_1 \leq \eps,$$and
	$$  \| \rho_{ t\mhyphen2\nmext(X,Y)t\mhyphen2\nmext(X^1,Y^1) \ldots t\mhyphen2\nmext(X^t,Y^t) YY^{[t]}M'} - U_m \otimes \rho_{t\mhyphen2\nmext(X^1,Y^1) \ldots t\mhyphen2\nmext(X^t,Y^t) YY^{[t]}M'} \|_1 \leq \eps.$$
\end{definition}

\begin{definition}[Quantum one-many split-state adversary]\label{def:splitstateadvt}Let $\sigma_{SXY}$ be the state after encoding message $S$.  The quantum one-many split-state adversary (denoted $\mathcal{A} = (U,V,{\psi}_{})$) will act via two isometries, $(U,V)$ using an additional shared entangled state $\ket{\psi}_{NM}$ as specified by $U: \cH_X \otimes \cH_N \rightarrow \cH_{X^{[t]}}  \otimes \cH_{N^\prime}$ and $V: \cH_Y \otimes \cH_M \rightarrow  \cH_{Y^{[t]}} \otimes \cH_{M'}$. Let $\hat{\rho} = (U \otimes V)(\sigma \otimes \ket{\psi}_{}\bra{\psi}_{})(U \otimes V)^\dagger$ and ${\rho}$ be the final state after measuring the registers $(X^{[t]}Y^{[t]})$ in the computational basis~\footnote{We enforce the quantum split-state adversary to return classical registers $X^{[t]},Y^{[t]}$ by doing the measurement in the computational basis. This is without any loss of generality since the decoding process can start by doing the measurement (this will not affect anything in the absence of quantum one-many split-state adversary) and then decode using $\dec$.}. 
\end{definition}

\begin{definition}[Quantum secure one-many non-malleable codes in the split-state model\label{def:nmcodes_qst}]  An encoding and decoding scheme $(\enc, \dec)$ is a $(t;m,n,\eps)$-quantum secure one-many non-malleable code in the split-state model with error $\eps$, if for state $\rho$ and adversary $\mathcal{A} = (U,V, {\psi})$ (as defined in Definition~\ref{def:splitstateadvt}), there exists a random variable $  \mathcal{D}_{\mathcal{A}}$ on $(\{0,1 \}^m  \cup \lbrace \sm \rbrace)^t$ such that 
 $$\forall s \in \{0,1\}^m: \qquad \Vert S^{[t]}_s- \cp^{(t)} (\mathcal{D}_{\mathcal{A}} ,s)   \Vert_1  \leq \eps.$$
 Above $\forall i \in [t], S^i=\dec(X^i,Y^i)$ and $ S^{[t]}_s = (S^{[t]}|S=s)$.
\end{definition}

\begin{claim}
\label{factcor:2nmextextt}
 Let $\sigma_{X\hat{X}NY\hat{Y}M}$ be a $(k_1,k_2)\mhyphen \qmas$ with $\vert X \vert=n$ and $\vert Y \vert=n$. Let $t \mhyphen2\nmext:\{0,1\}^n \times\{0,1\}^{n} \to \lbrace 0,1\rbrace^{m}$ be a $(t;k_1,k_2, \eps)$-quantum secure $2$-source $t \mhyphen$non-malleable extractor. Let $S=t \mhyphen 2\nmext(X,Y)$. Then,
$$ \| \sigma_{SYM} - U_{m} \otimes \sigma_{YM} \|_1 \leq \eps.$$
\end{claim}
\begin{proof}
Let $U: \cH_X \rightarrow \cH_X \otimes \cH_{X^{[t]}} \otimes  \cH_{\hat{X}^{[t]}}$, $V: \cH_Y \rightarrow \cH_Y \otimes \cH_{Y^{[t]}} \otimes  \cH_{\hat{Y}^{[t]}}$ be isometries such that for $\rho = (U \otimes V) \sigma (U \otimes V)^\dagger$, we have $X^{[t]}Y^{[t]}$ classical (with copies $\hat{X}^{[t]}\hat{Y}^{[t]}$ respectively) and for every $i \in [t]$ either  $\Pr(X \ne X^i)_\rho =1$ or $\Pr(Y \ne Y^i)_\rho =1.$~\footnote{It is easily seen that such isometries exists.} Notice the state $\rho$ is a $(t;k_1,k_2)\mhyphen\nmas$. Since $t \mhyphen2\nmext$ is a  $(t;k_1,k_2, \eps)$-quantum secure (see Definition~\ref{def:2tsourcenme}), we have $$ \| \rho_{SS^{[t]}YY^{[t]}M} - U_{m} \otimes \rho_{S^{[t]}YY^{[t]}M} \|_1 \leq \eps,$$  where $S^i =t \mhyphen2\nmext(X^i,Y^i)$. Using Fact~\ref{fact:data}, we get $$ \| \rho_{SYM} - U_{m} \otimes \rho_{YM} \|_1 \leq \eps.$$ 
The desired now follows by noting $\sigma_{XNMY} =  \rho_{XNMY}.$
\end{proof}

\begin{claim}
\label{factcor:2nmextexttt}
Let $t \mhyphen2\nmext:\{0,1\}^n \times\{0,1\}^{n} \to \lbrace 0,1\rbrace^{m}$ be a  $(t;k_1,k_2, \eps)$-quantum secure $2$-source $t \mhyphen$non-malleable extractor. Then, $t \mhyphen2\nmext$ is also a $(j;k_1,k_2, \eps)$-quantum secure $2$-source $j \mhyphen$non-malleable extractor for any positive integer $j \leq t.$
 
\end{claim}
\begin{proof}
Let $\sigma_{X\hat{X}NY\hat{Y}M}$ be a $(k_1,k_2)\mhyphen \qmas$ with $\vert X \vert=n$ and $\vert Y \vert=n$. Let $U: \cH_X \otimes  \cH_N \rightarrow \cH_X \otimes \cH_{X^{[j]}} \otimes  \cH_{\hat{X}^{[j]}} \otimes  \cH_{N'}$, $V: \cH_Y \otimes  \cH_M \rightarrow \cH_Y \otimes \cH_{Y^{[j]}} \otimes  \cH_{\hat{Y}^{[j]}} \otimes  \cH_{M'}$ be isometries such that for $\rho = (U \otimes V) \sigma (U \otimes V)^\dagger$, we have $X^{[j]}Y^{[j]}$ classical (with copies $\hat{X}^{[j]}\hat{Y}^{[j]}$ respectively) and for every $i \in [j]$ either  $\Pr(X \ne X^i)_\rho =1$ or $\Pr(Y \ne Y^i)_\rho =1.$ Notice the state $\rho$ is a $(j;k_1,k_2)\mhyphen\nmas$. 

Let $U': \cH_X \rightarrow \cH_X \otimes \cH_{X^{[t]\setminus[j]}} \otimes  \cH_{\hat{X}^{[t]\setminus[j]}}$, $V': \cH_Y \rightarrow \cH_Y \otimes \cH_{Y^{[t]\setminus[j]}} \otimes  \cH_{\hat{Y}^{[t]\setminus[j]}}$ be  isometries~\footnote{It is easily seen that such isometries exists.} such that for $\rho' = (U' \otimes V') \rho (U' \otimes V')^\dagger$, we have $X^{[t]\setminus[j]}Y^{[t]\setminus[j]}$ classical (with copies $\hat{X}^{[t]\setminus[j]}\hat{Y}^{[t]\setminus[j]}$ respectively) and for every $i \in [t]$ either  $\Pr(X \ne X^i)_{\rho'} =1$ or $\Pr(Y \ne Y^i)_{\rho'} =1.$ Notice the state $\rho'$ is a $(t;k_1,k_2)\mhyphen\nmas$. Since $t \mhyphen2\nmext$ is a  $(t;k_1,k_2, \eps)$-quantum secure (see Definition~\ref{def:2tsourcenme}), we have
\[ \| \rho'_{SS^{[t]}XX^{[t]}M'} - U_{m} \otimes \rho'_{S^{[t]}XX^{[t]}M'} \|_1 \leq \eps \quad ;\quad  \| \rho'_{SS^{[t]}YY^{[t]}M'} - U_{m} \otimes \rho'_{S^{[t]}YY^{[t]}M'} \|_1 \leq \eps,\]
 where $S^i =t \mhyphen2\nmext(X^i,Y^i)$. Using Fact~\ref{fact:data}, we get 
 \[ \| \rho'_{SS^{[j]}XX^{[j]}M'} - U_{m} \otimes \rho'_{S^{[j]}XX^{[j]}M'} \|_1 \leq \eps \quad ;\quad  \| \rho'_{SS^{[j]}YY^{[j]}M'} - U_{m} \otimes \rho'_{S^{[j]}YY^{[j]}M'} \|_1 \leq \eps.\]Noting $\rho'_{XX^{[j]}N'M'YY^{[j]}} =  \rho_{XX^{[j]}N'M'YY^{[j]}},$ we finally get
\[ \| \rho_{SS^{[j]}XX^{[j]}M'} - U_{m} \otimes \rho_{S^{[j]}XX^{[j]}M'} \|_1 \leq \eps \quad ;\quad  \| \rho_{SS^{[j]}YY^{[j]}M'} - U_{m} \otimes \rho_{S^{[j]}YY^{[j]}M'} \|_1 \leq \eps.\]

\end{proof}

\begin{theorem}[Quantum secure one-many  non-malleable code in the split-state model]\label{lem:qnmcodesfromnmextt} Let $t \mhyphen 2\nmext : \{0,1\}^{n} \times \{0,1\}^{n} \to \{0,1\}^m$ be  a  $(t;n-k,n-k, \eps)$-quantum secure $2$-source $t \mhyphen$non-malleable extractor. There exists a $(t;m,n,\eps')$-quantum secure one-many non-malleable code with parameter $\eps' = 2^m((2^{-k} + \eps)2^{2t} +  \eps)$.
\end{theorem}
\begin{proof}

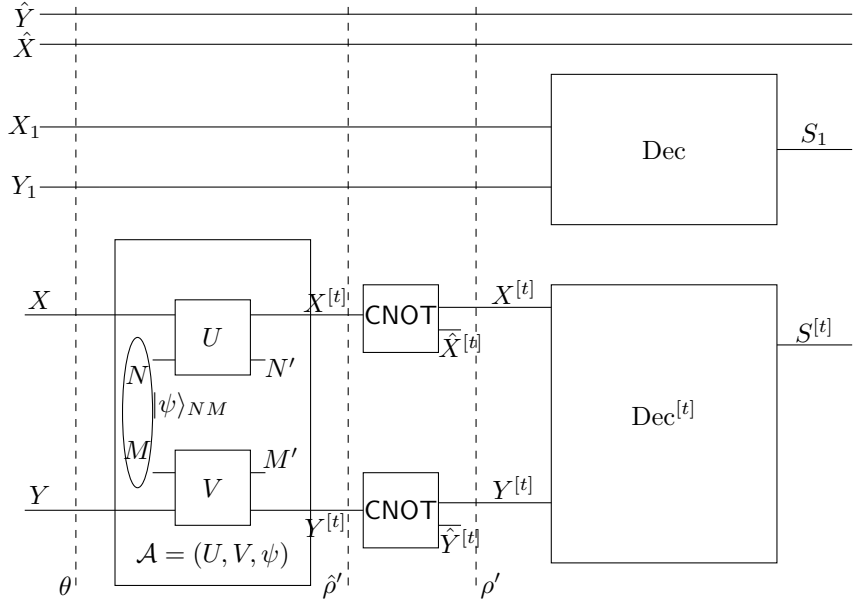
\begin{figure}
\centering
\begin{tikzpicture}


\draw (4.2,6.4) -- (15,6.4);
\node at (4,6.4) {$\hat{X}$};

\draw (4.2,6.8) -- (15,6.8);
\node at (4,6.8) {$\hat{Y}$};



\draw (4.2,5.3) -- (11,5.3);
\node at (4,5.3) {$X_1$};

\node at (4,4.5) {$Y_1$};
\draw (4.2,4.5) -- (5,4.5);
\draw (5,4.5) -- (11,4.5);
\node at (14.5,5.2) {$S_1$};
\draw (14,5) -- (15,5);

\draw (11,4) rectangle (14,6);
\node at (12.5,5) {$\dec$};


\draw [dashed] (4.68,-0.8) -- (4.68,7);
\draw [dashed] (8.3,-0.8) -- (8.3,7);
\draw [dashed] (10,-0.8) -- (10,7);

\node at (6.2,1.6) {$\ket{\psi}_{NM}$};
\node at (4.54,-0.8) {$\theta$};
\node at (8.1,-0.8) {$\hat{\rho}'$};
\node at (10.2,-0.8) {${\rho}'$};

\node at (4.2,3) {$X$};
\node at (8,3) {$X^{[t]}$};
\node at (10.5,3.1) {$X^{[t]}$};
\node at (9.8,2.4) {$\hat{X}^{[t]}$};
\draw (4,2.8) -- (6,2.8);
\draw (7,2.8) -- (8.5,2.8);
\draw (9.5,2.9) -- (11,2.9);
\draw (9.5,2.6) -- (9.8,2.6);
\node at (14.5,2.6) {$S^{[t]}$};
\draw (14,2.4) -- (15,2.4);

\node at (4.2,0.4) {$Y$};
\node at (8,0.0) {$Y^{[t]}$};
\draw (4,0.2) -- (6,0.2);
\draw (7,0.2) -- (8.5,0.2);
\draw (9.5,0.3) -- (11,0.3);
\draw (9.5,0) -- (9.8,0);
\node at (10.5,0.5) {$Y^{[t]}$};
\node at (9.8,-0.2) {$\hat{Y}^{[t]}$};

\draw (6,2) rectangle (7,3);
\node at (6.5,2.5) {$U$};
\draw (6,0) rectangle (7,1);
\node at (6.5,0.5) {$V$};

\node at (6.5,-0.4) {$\mathcal{A}=(U,V,\psi)$};
\draw (5.2,-0.8) rectangle (7.8,3.8);

\draw (5.5,1.5) ellipse (0.2cm and 1cm);
\node at (5.5,2) {$N$};
\draw (5.7,2.2) -- (6,2.2);
\node at (7.4,2.1) {$N'$};
\draw (7,2.2) -- (7.2,2.2);
\node at (5.5,1) {$M$};
\draw (5.7,0.7) -- (6,0.7);
\node at (7.4,0.9) {$M'$};
\draw (7.0,0.7) -- (7.2,0.7);

\draw (8.5,2.3) rectangle (9.5,3.2);
\node at (9,2.8) {$\mathsf{CNOT}$};

\draw (8.5,-0.3) rectangle (9.5,0.7);
\node at (9,0.2) {$\mathsf{CNOT}$};

\draw (11,-0.5) rectangle (14,3.2);
\node at (12.5,1.5) {$\dec^{[t]}$};

\end{tikzpicture}
\caption{Analysis of a quantum secure one-many non-malleable code in the split-state model}\label{fig:splitstate2t}
\end{figure}

The proof proceeds in similar lines of Theorem~\ref{lem:qnmcodesfromnmext}. We do not repeat the entire argument but provide the necessary details required to complete the proof. 

Let $\mathcal{A}=(U,V,{\psi})$ be the quantum split-state adversary from Definition~\ref{def:splitstateadvt}. We show that the encoding based on $t \mhyphen 2\nmext$ is a  $(t;m,n, \eps')$-quantum secure one-many non-malleable code. Using arguments similar to Theorem~\ref{lem:qnmcodesfromnmext}, it suffices to prove 
\[  \Vert \rho'_{ S_1S^{[t]}} - U_m  \cp^{(t)} (\mathcal{D}_{\mathcal{A}} ,U_m)   \Vert_1  \leq (2^{-k} + \eps)2^{2t}, \] 
where $\rho'$ is the pure state as in Figure~\ref{fig:splitstate2t} after the action of adversary $\mathcal{A}$ on state $\theta$. Note $S^i = t \mhyphen 2\nmext(X^i,Y^i)$ and state   $\theta_{X\hat{X}X_1 Y\hat{Y}Y_1 } = \theta_{X\hat{X}X_1 } \otimes \theta_{Y\hat{Y}Y_1 }$ is a pure state such that $\theta_X = \theta_Y =U_n$, ($X_1, \hat{X}$) are copies of $X$, ($Y_1, \hat{Y}$) are copies of $Y$ respectively.~\footnote{It is easily seen such state exists.} In the Figure~\ref{fig:splitstate2t}, with some abuse of notation, we used $\dec^{[t]}$ to denote $\dec (X^i,Y^i)=S^i$ performed for every $i \in [t]$.

Let $\rho'$ be the final pure state after we generate $t$-bit classical registers $C,D$ (with copies $\hat{C},\hat{D}$ respectively) such that $C_i=1$ indicates $X_1 \ne X^{i}$ in state ${\rho}'$ and $D_i=1$ indicates $Y_1 \ne Y^i$ in state ${\rho}'$ for every $i \in [t]$.

Using similar arguments as in  Theorem~\ref{lem:qnmcodesfromnmext}, one can note that the state $\rho'$ is an  $(n,n)$-$\qmas$. For $C=c \in \{ 0,1\}^t,D=d \in \{ 0,1\}^t$, denote $\mathcal{S}^{c,d} = \{  i \in [t] : (c_i=1) \vee  (d_i=1)\}$ and  $\rho'^{c,d} = \rho' \vert ((C,D)=(c,d))$. 

For every $c,d \in \{ 0,1 \}^t$, let  $\mathcal{D}^{c,d}_{\mathcal{A}}=\mathcal{D}^{c,d}_{(U,V,{\psi})} = (\mathcal{D}^{c,d}_{\mathcal{A},1}, \ldots,  \mathcal{D}^{c,d}_{\mathcal{A},t}),$ where 
$\mathcal{D}^{c,d}_{\mathcal{A},j} = \rho'^{c,d}_{t \mhyphen 2\nmext(X^j,Y^j) } $ for $j \in \mathcal{S}^{c,d}$ and  $\mathcal{D}^{c,d}_{\mathcal{A},j}  $ be the distribution that is deterministically equal to $\sm $ otherwise. Let $\mathcal{D}_{\mathcal{A}} = \sum_{c,d \in \{ 0,1 \}^t} \Pr((C,D)=(c,d))_{\rho'} \mathcal{D}^{c,d}_{\mathcal{A}}$. Note for every $c,d \in \{ 0,1 \}^t$, the value $\Pr((C,D)=(c,d))_{\rho'}$ depends only on $(U,V,{\psi})$. Note
\begin{equation}\label{eq:newproof71}
    \rho'_{S_1S^{[t]}} = \sum_{c,d \in \{ 0,1\}^t}\Pr((C,D)=(c,d))_{\rho'} \rho'^{c,d}_{S_1S^{[t]}},
\end{equation}and \begin{equation}\label{eq:newproof72}
     Z  \cp^{(t)} (\mathcal{D}_{\mathcal{A}} ,Z)= \sum_{c,d \in \{ 0,1\}^t}\Pr((C,D)=(c,d))_{\rho'} Z\cp^{(t)}(\mathcal{D}^{c,d}_{\mathcal{A}} ,Z) ,
\end{equation}
where $Z=U_m.$

\begin{claim}\label{claim:987t}For every $c,d \in \{ 0,1 \}^t$, we have
\[\Pr((C,D)=(c,d))_{\rho'} \Vert   \rho'^{c,d}_{S_1S^{[t]}}  
- Z  \cp^{(t)}(\mathcal{D}^{c,d}_{\mathcal{A}} ,Z)  \Vert_1 \leq 2^{-k} + \eps ,  \]
where $Z=U_m$.
\end{claim}
\begin{proof}

Fix $c,d \in \{0,1\}^t$. Suppose $\Pr((C,D)=(c,d))_{\rho'}  \leq 2^{-k}$, then we are done. Thus we assume otherwise. Note $\mathcal{S}^{c,d}$ is empty only when $(c,d)=(0^t,0^t)$.

First let $\mathcal{S}^{c,d} = \{ i_1, \ldots i_j\}$ be non-empty. Using arguments similar to Claim~\ref{claim:987} involving Fact~\ref{fact:minentropydecrease} and noting that state $\rho'$ is an $(n,n)\mhyphen\qmas$, we get that $\rho'^{c,d}$ is a $(j;n-k,n-k)\mhyphen\nmas$. Using  Claim~\ref{factcor:2nmextexttt} to note $t \mhyphen2\nmext$ is also  $(j;n-k,n-k,\eps)$-quantum secure for any positive integer $j \leq t$ and using  Fact~\ref{fact:data}, we have 
\begin{equation}\label{eq:equationlish}
    \Vert   \rho'^{c,d}_{S_1S^{i_1} \ldots S^{i_j}}  
- U_m \otimes  \rho'^{c,d}_{S_1S^{i_1} \ldots S^{i_j}}  \Vert_1 \leq  \eps .
\end{equation}
For any $p \in [t]\setminus \mathcal{S}^{c,d}$, we have 
$  \rho'^{c,d}_{S_1}= \rho'^{c,d}_{S^p}.$ Thus from Eq.~\eqref{eq:equationlish} and using Fact~\ref{fact:data}, we have 
$$\Vert   \rho'^{c,d}_{S_1S^{[t]}}  
- Z  \cp^{(t)}(\mathcal{D}^{c,d}_{\mathcal{A}} ,Z)  \Vert_1 \leq  \eps.$$

For $(c,d)=(0^t,0^t)$, the proof follows noting $\rho'^{c,d}$ is an  $(n-k,n-k)\mhyphen\qmas$ and using Claim~\ref{factcor:2nmextextt} along with Fact~\ref{fact:data}.
\end{proof}Consider,
\begin{align*}
   & \Vert \rho'_{S_1S^{[t]}} - Z  \cp^{(t)} (\mathcal{D}_{\mathcal{A}} ,Z)   \Vert_1 \\ 
    & =\Vert \sum_{c,d \in \{ 0,1\}^t}\Pr((C,D)=(c,d))_{\rho'} \rho'^{c,d}_{S_1S^{[t]}} - Z \cp^{(t)} (\mathcal{D}_{\mathcal{A}},Z)   \Vert_1 & \mbox{(Eq.~\eqref{eq:newproof71})}\\ 
   &\leq
    \sum_{c,d \in \{ 0,1\}^t}\Pr((C,D)=(c,d))_{\rho'} \Vert \rho'^{c,d}_{S_1S^{[t]}} - Z \cp^{(t)} (\mathcal{D}^{c,d}_{\mathcal{A}} ,Z)  \Vert_1 & \mbox{(Eq.~\eqref{eq:newproof72} and Fact~\ref{fact:traceconvex})} \\
   & \leq (2^{-k} + \eps)2^{2t} & \mbox{(Claim~\ref{claim:987t})}.
\end{align*} 
This completes the proof.
\end{proof}

\subsection{A quantum secure modified $2$-source $t$-non-malleable extractor\label{sec:2tnm}}

 We modify the construction of $t$-$2\nmext$ from~\cite{BJK21}, using ideas from~\cite{CGL15} to construct a $\mathsf{new}$-$t$-$2\nmext : \{0,1\}^{n} \times \{0,1\}^{n} \to \{0,1\}^m$ that is  $(n-n_1,n-n_1, \cO(\eps))$-secure against $\nma$ for parameters $n_1=n^{\Omega(1)}$, $m= n^{1-\Omega(1)}$,  $t=n^{\Omega(1)}$ and $\eps = 2^{-n^{\Omega(1)}}$.
 \subsection*{Parameters}\label{sec:parameters_2t_sourcenm}
Let $\delta, \delta_1, \delta_2, \delta_3 >0$ be small enough constants such that $\delta_1 < \delta_2$. Let  $n,n_1,n_2,n_3,n_4,n_5,n_6,n_7,n_x,n_y,a,s,b,h,t$ be positive integers and
 $\eps', \eps > 0$ such that: 
  \[  n_1 = n^{\delta_2} \quad ; \quad  n_2 =n-3n_1  \quad ; \quad q= 2^{\log \left(n+1\right) }  \quad ; \quad \eps= 2^{- \cO(n^{\delta_1})} \quad ;\]
 
\[ n_3 = \frac{n_1}{10} \quad ; \quad n_4 = \frac{n_2}{\log (n+1)} \quad ; \quad n_5 = n^{\delta_2/3} \quad ;\quad a=6n_1+2 n_5 \log (n+1) = \mathcal{O}(n_1) \quad ;     \]

\[ n_6 = 3n_1^3 \quad ; \quad n_7 = n-3n_1-n_6  \quad ; \quad n_x = \frac{n_7}{12a} \quad ; \quad n_y = \frac{n_7}{12a} \quad ; \quad  2^{\cO(a)}\sqrt{\eps'} = \eps \quad ;  \]

\[s = \cO\left(\log^2\left(\frac{n}{\eps'}\right)\log n \right)  \quad  ; \quad b = \cO\left( \log^2\left(\frac{n}{\eps'}\right) \log n \right) \quad ; \quad  t \leq n^{\delta_3} \quad ; \quad h = 10ts \quad \]
\begin{itemize}\label{sec:extparameters_2nmt}
    \item $\IP_1$ be $\IP^{3n_1/n_3}_{2^{n_3}}$,
    \item $\Ext_1$ be $(2b, \eps')$-quantum secure $(n_y,s,b)$-extractor,
    \item $\Ext_2$ be $(2s, \eps')$-quantum secure $(h,b,s)$-extractor,
    \item $\Ext_3$ be $(4h, \eps')$-quantum secure $(n_x,b,2h)$-extractor,
    \item $\Ext_4$ be $(n_y/4t, \eps^2)$-quantum secure $(4n_y,2h,n_y/8t)$-extractor,
    \item $\IP_2$ be $\IP^{3n_1^3/2h}_{2^{2h}}$,
    \item  $\Ext_6$ be $(\frac{n_x}{2t}, \eps^2)$-quantum secure $(4n_x,n_y/8t,n_x/4t)$-extractor. 
 \end{itemize}

 \subsection*{Definition of $2$-source $t$-non-malleable extractor}

\begin{algorithm}
\caption{: $\mathsf{new}$-$t\mhyphen2\nmext: \lbrace 0,1 \rbrace ^n\times \lbrace 0,1 \rbrace^n   \rightarrow \lbrace 0,1 \rbrace^{n_x/4t}$}\label{alg:2tnmExt}
\begin{algorithmic}
\State{}

\noindent \textbf{ Input:}  $X, Y$\\


\begin{enumerate}
    \item Advice generator: \[X_1=\pre(X,1,3n_1) \quad  ; \quad Y_1 = \pre(Y,1,3n_1) \quad ;
    \quad R= \IP_1(X_1,Y_1) \quad ; \quad  \]
    \[X_2=\pre(X,3n_1+1,n) \quad  ; \quad Y_2 = \pre(Y,3n_1+1,n) \quad ;\]
    \[ \quad G=X_1 \circ \bar{X}_2 \circ Y_1 \circ \bar{Y}_2=X_1 \circ \ecc(X_2)_{\samp(R)} \circ Y_1 \circ \ecc(Y_2)_{\samp(R)} \]

    \item $X_3=\pre(X,3n_1 +1,3n_1+n_6) \quad ; \quad Y_3=\pre(Y,3n_1 +1,3n_1+n_6)\quad; \quad Z^1=\IP_2(X_3,Y_3)$
     \item $X_4=\pre(X,3n_1+n_6+1,n) \quad ; \quad Y_4=\pre(Y,3n_1+n_6+1,n)\quad$
    \item Correlation breaker with advice:\quad $F=2\advcb(Y_4,X_4,Z^1, G)$
    \item $X^{a+1}=\pre(X_4,4n_xa+1,4n_xa+4n_x)$
    \item $S=\Ext_6(X^{a+1},F)$
\end{enumerate}

 \noindent \textbf{ Output:} $S$ 
\end{algorithmic}
\end{algorithm}

 The following theorem shows that the function $\mathsf{new}$-$t\mhyphen2\nmext$ as defined in Algorithm~\ref{alg:2tnmExt} is $(t;n-n_1,n-n_1,\cO(\eps))$-secure against $\nma$ by noting that  $S=\mathsf{new}\mhyphen t\mhyphen2\nmext(X,Y)$ and $S^i=\mathsf{new}\mhyphen t\mhyphen2\nmext(X^i,Y^i)$ for every $i \in [t]$.
 Note that $2\advcb$ in Algorithm~\ref{alg:2tnmExt} is same as the one in Algorithm~\ref{alg:2AdvCB} except for parameters and extractors which are to be used as mentioned in this section.  

\begin{theorem}[Security of $\mathsf{new}$-$t\mhyphen2\nmext$]\label{thm:nmext2t}
Let $\rho_{{X X^{[t]} \hat{X} \hat{X}^{[t]} N YY^{[t]} \hat{Y} \hat{Y}^{[t]} M}}$ be a $(t;n-n_1,n-n_1)\mhyphen\nmas$. Then,
\[
   \Vert \rho_{ S S^{[t]} Y  Y^{[t]} M} - U_{n_x/4t} \otimes \rho_{S^{[t]} Y  Y^{[t]} M} \Vert_1 \leq \cO(\eps).
\] 
\end{theorem}
\begin{proof}
The proof proceeds in similar lines to the proof of Theorem $9$ in~\cite{BJK21} using Fact~\ref{lem:2} for alternating extraction argument, Fact~\ref{lem:minentropy} for bounding the min-entropy required in alternating extraction and Facts~\ref{l-qma-needed-fact},~\ref{lemma:nearby_rho_prime_prime} for the security of inner-product function in $(k_1,k_2) \mhyphen \qmas$ framework, Fact~\ref{fidelty_trace} for relation between $\Delta_B, \Delta$ and we do not repeat it.
\end{proof}

\subsection*{Efficiently sampling from the preimage of $\mathsf{new}$-$t\mhyphen2\nmext$}

\begin{theorem}\label{thm:main1}
Let $\mathsf{new}$-$t\mhyphen2\nmext: \lbrace 0,1 \rbrace ^n\times \lbrace 0,1 \rbrace^n   \rightarrow \lbrace 0,1 \rbrace^{n_x/4t}$ be the $2$-source $t$-non-malleable extractor from Algorithm~\ref{alg:2tnmExt}. Let $(\enc,\dec)$ be the encoding and decoding based on $\mathsf{new}$-$t \mhyphen 2\nmext$. Let $\hat{S}\enc(\hat{S})=\hat{S}\hat{X}\hat{Y}$ for a uniform message $\hat{S}=U_{n_x/4t}$. There exists an efficient algorithm that can sample from a distribution $\tilde{S}\tilde{X}\tilde{Y}$ such that $\Vert \tilde{S}\tilde{X}\tilde{Y} -\hat{S}\hat{X}\hat{Y} \Vert_1 \leq \cO(\eps)$ and $\tilde{S}= U_{n_x/4t}$.

\end{theorem}
\begin{proof}
The proof follows similarly to that of Theorem~\ref{thm:main} and we do not repeat it.
\end{proof}

We have the following corollary. 
\begin{corollary} Let $0<\delta_4 + \delta_3 < \delta_1$ and $m'=n^{\delta_4}$ be an integer. Let $\mathsf{new}$-$t \mhyphen 2\nmext: \lbrace 0,1 \rbrace ^n\times \lbrace 0,1 \rbrace^n   \rightarrow \lbrace 0,1 \rbrace^{m}$ be the quantum secure $2$-source $t$-non-malleable extractor from Algorithm~\ref{alg:2tnmExt}, where  $m=n_x/4t$. Let $2\nmext: \lbrace 0,1 \rbrace ^n\times \lbrace 0,1 \rbrace^n   \rightarrow \lbrace 0,1 \rbrace^{m'}$ be such that $2\nmext(X,Y)$ is same as $\mathsf{new}$-$t \mhyphen 2\nmext (X,Y)$ truncated to first $m'$ bits.  Let $(\enc,\dec)$ be the encoding and decoding based on $2\nmext$. Let $\hat{S}\enc(\hat{S})=\hat{S}\hat{X}\hat{Y}$ for a uniform message $\hat{S}=U_{m'}$. There exists an efficient algorithm that can sample from a distribution $\tilde{S}\tilde{X}\tilde{Y}$ such that $\tilde{S}= U_{m'}$ and for every $s \in \{0,1 \}^{m'}$, we have $\Vert (\tilde{X}\tilde{Y}) \vert (\tilde{S}=s) -(\hat{X}\hat{Y}) \vert (\hat{S}=s) \Vert_1 \leq \cO(2^{m'}\eps) \leq 2^{-n^{\Omega(1)}}$. 
\end{corollary}

\section{A quantum secure  non-malleable secret sharing scheme\label{Appendixss}}

Secret sharing is a fundamental primitive in cryptography where a dealer encodes a secret/message into shares and distributes among many parties. Only the authorized subsets of parties should be able to recover the initial secret. Most well known secret sharing schemes are the so called $t$-out-of-$n$ secret sharing schemes where at least $t$-parties are required to decode the secret~\cite{Sha79,Bla79}. In this paper, we focus only on $2$-out-of-$2$ secret sharing schemes. 

Recently non-malleable secret sharing schemes are introduced by Goyal and Kumar~\cite{GK16} with the additional guarantee that when the adversary tampers with possibly all the shares of the secret independently, then the reconstruction procedure outputs original secret or something that is unrelated to the original secret. 

In this paper, in addition, we allow the adversary to make use of arbitrary entanglement to tamper the shares. We then require the reconstruction procedure to output original secret or something that is unrelated to the original secret. We call such secret sharing schemes as quantum secure non-malleable secret sharing schemes. We show that quantum secure non-malleable codes in the split-state model gives rise to quantum secure $2$-out-of-$2$ non-malleable secret sharing schemes.

We consider an encoding and decoding scheme $(\enc,\dec)$  where $\enc(S)=(X,Y)$. Here $S \sim U_m$ ($U_m$ is uniform distribution on $m$ bits) represents the secret/message and  $X,Y\in\{0,1\}^n$ are the two shares/parts of the codeword. $\enc$ is a randomized function and $\dec(X,Y)$ is a deterministic function, such that $\Pr\left(\dec\left( \enc(S)\right) =S \right)=1$.  Let $\mathcal{F}$ denote the set of all functions $f : \{0,1 \}^n \to \{0,1 \}^n$.

\begin{definition}[$2$-out-of-$2$ non-malleable secret sharing scheme\label{def:nmsscodes}~\cite{GK16}]$(\enc, \dec)$ is an $(m,n,\eps_1, \eps_2)$-$2$-out-of-$2$ non-malleable secret sharing scheme, if

\begin{itemize}
    \item \textbf{statistical privacy:} for any two  secrets $s_1,s_2 \in \{0,1 \}^m$, it holds that
\[ \Vert X_{s_1} -X_{s_2} \Vert_1 \leq \eps_1 \quad ; \quad \Vert Y_{s_1} -Y_{s_2} \Vert_1 \leq \eps_1,\] 
where $(X_{s_1},Y_{s_1})=\enc(s_1)$ and $(X_{s_2},Y_{s_2})=\enc(s_2)$
\item \textbf{non-malleability:} for every $f=(g_1,g_2) \in \mathcal{F} \times \mathcal{F}$, there exists a random variable $\mathcal{D}_{f}=\mathcal{D}_{(g_1,g_2)}$ on $\{0,1 \}^m  \cup \lbrace \sm \rbrace$ which is independent of the randomness in $\enc$ such that for all secrets $s \in \{0,1 \}^m$, it holds that
$$  \Vert  \dec(f (\enc(s)))  - \cp (\mathcal{D}_{f} ,s)\Vert_1 \leq \eps_2,$$
 where the function $\cp(x,s)$ is such that $\cp(x,s) =s$ if $x=\mathsf{same}$, otherwise $\cp(x,s) =x$. 
\end{itemize}

\end{definition}

\begin{definition}[Quantum secure $2$-out-of-$2$ non-malleable secret sharing scheme\label{def:qnmsscodes}]  $(\enc, \dec)$ is an $(m,n,\eps_1,\eps_2)$-quantum secure $2$-out-of-$2$  non-malleable secret sharing scheme, if 
\begin{itemize}
    \item \textbf{statistical privacy:} is as defined in Definition~\ref{def:nmsscodes}
\item \textbf{non-malleablity:}
for state $\rho$ and adversary $\mathcal{A} = (U,V, {\psi})$ (as defined in Definition~\ref{def:splitstateadv}), there exists a random variable $\mathcal{D}_{\mathcal{A}}$~\footnote{Distribution depends only on $\mathcal{A}$ and is independent of the original secret $S$.} on $\{0,1 \}^m  \cup \lbrace \sm \rbrace$ such that 
 $$\forall s \in \{0,1\}^m: \qquad \Vert S'_s- \cp (\mathcal{D}_{\mathcal{A}} ,s)   \Vert_1  \leq \eps_2.$$Above $S'=\dec(X',Y')$,  $S'_s = (S'|S=s)$ and the function $\cp$ is as defined in Definition~\ref{def:nmsscodes}.
\end{itemize}

\end{definition}
\begin{fact}[Lemma 6.1 in~\cite{ADKO17}]\label{fact:ssprivacy}
Let $(\enc,\dec)$ be an $(m,n,\eps)$-non-malleable code in the split-state model. Then, 
 for any two  messages $u,v \in \{0,1 \}^m$, it holds that
\[ \Vert X_{u} -X_{v} \Vert_1 \leq 2\eps \quad ; \quad \Vert Y_{u} -Y_{v} \Vert_1 \leq 2\eps,\] 
where $(X_{u},Y_{u})=\enc(u)$ and $(X_{v},Y_{v})=\enc(v)$. 

\end{fact}
\begin{remark}
Since a quantum secure non-malleable code is also a classical non-malleable code in the split-state model, the above fact also applies for quantum secure non-malleable code in the split-state model.
\end{remark}
\begin{corollary}
Let $(\enc,\dec)$ be an $(m,n,\eps)$-quantum secure non-malleable code in the split-state model. Then, $(\enc,\dec)$ is also an $(m,n,2\eps, \eps)$-quantum secure $2$-out-of-$2$ non-malleable secret sharing scheme.
\end{corollary}
\begin{proof}
Statistical privacy follows from the  Fact~\ref{fact:ssprivacy}. Non-malleability property follows from the Definition~\ref{def:nmcodes_qs} of quantum secure non-malleable code in the split-state model.
\end{proof}

\end{document}